\documentclass[10pt,a4paper]{article}
\pdfoutput=1
\usepackage{amsmath}
\usepackage{amsfonts}
\usepackage{amssymb}
\usepackage{amscd}
\usepackage{amsthm}
\usepackage{latexsym}
\usepackage[all]{xy}
\usepackage{url}

\newcommand{\CDL}{\mathsf{CDL}}
\newcommand{\PDL}{\mathsf{PDL}}
\newcommand{\GL}{\mathsf{GL}}

\newtheorem{theorem}{Theorem}
\newtheorem{proposition}{Proposition}
\newtheorem{lemma}{Lemma}
\newtheorem{corollary}{Corollary}
\newtheorem{example}{Example}

\begin{document}

\markboth{}{}
\title{Concurrent Dynamic Algebra}
\author{Hitoshi Furusawa\\Kagoshima University,\\ Japan \and  Georg Struth\\University of Sheffield,\\ United Kingdom}

\maketitle

\begin{abstract}
  We reconstruct Peleg's concurrent dynamic logic in the context of
  modal Kleene algebras. We explore the algebraic structure of its
  multirelational semantics and develop an abstract axiomatisation of
  concurrent dynamic algebras from that basis. In this axiomatisation,
  sequential composition is not associative. It interacts with
  concurrent composition through a weak distributivity law. The modal
  operators of concurrent dynamic algebra are obtained from abstract
  axioms for domain and antidomain operators; the Kleene star is
  modelled as a least fixpoint. Algebraic variants of Peleg's axioms
  are shown to be valid in these algebras and their soundness is
  proved relative to the multirelational model. Additional results
  include iteration principles for the Kleene star and a refutation of
  variants of Segerberg's axiom in the multirelational setting. The
  most important results have been verified formally with
  Isabelle/HOL.
\end{abstract}

%\category{}{}{}
%\terms{}
%\keywords{}
%\acmformat{}

%%%%%%%%%%%%%%%%%%%%%%%%%%%%%%%%%%%%%%%%%%%%%%%%%%%%%%%%%%%%%%%%%%%%%%%%%%%

\section{Introduction}\label{S:introduction}

Concurrent dynamic logic ($\CDL$) has been proposed almost three
decades ago by Peleg~\cite{Peleg87} as an extension of propositional
dynamic logic ($\PDL$)~\cite{HarelKozenTiuryn} to study concurrency in
``its purest form as the dual notion of nondeterminism''. In this
setting, a computational process is regarded as a tree with two dual
kinds of branchings. According to the first one, the process may
choose a transition along one of the possible branches. This is known
as angelic, internal or existential choice. According to the second
one, it progresses along all possible branches in parallel, which is
known as demonic, external or universal choice. This lends itself to a
number of interpretations.

One of them associates computations with games processes play against
a scheduler or environment as their opponent. A process wins if it can
successfully resolve all internal choices and respond to all external
choices enforced by the opponent. Another one considers machines which
accept inputs by nondeterministically choosing one transition along
exixtential branches and executing all transitions in parallel amongst
universal ones. In yet another one, universal choices correspond to
agents cooperating towards a collective goal while existential choices
are made in competition by individual agents. Finally, in
shared-variable concurrency, interferences caused by different threads
accessing a global variable are observed as nondeterministic
assignments by particular threads; hence as external choices imposed
by the other threads.

Historically, in fact, $\CDL$ has been influenced by work on
alternating state machines~\cite{ChandraKozenStockmeyer} and Parikh's
game logic ($\GL$)~\cite{Parikh83,Parikh85}, which is itself based on
$\PDL$. Other aspcets of concurrency such as communication or
synchronisation, which are at the heart of formalisms such as Petri
nets or process algebras, are ignored in its basic axiomatisation.

Standard $\PDL$ has a relational semantics. This captures the
input/output dependencies of sequential programs. Internal choice is
modelled as union, sequential composition as relational composition.
External choice, however, cannot be represented by this semantics. It
requires relating an individual input to a set of outputs, that is,
relations of type $A\times 2^B$ instead of $A\times B$. These are known
as \emph{multirelations}.

In multirelational semantics, external choice still corresponds to
union, but sequential composition must be redefined. According to
Parikh's definition, a pair $(a,A)$ is in the sequential composition
of multirelation $R$ with multirelation $S$ if $R$ relates element $a$
with an intermediate set $B$ and every element of $B$ is related to
the set $A$ by $S$. According to Peleg's more general definition, it
suffices that $S$ relates each element $b\in B$ with a set $C_b$ as
long as the union of all the sets $C_b$ yields the set $A$. In
addition, a notion of external choice or parallel composition can now
be defined. If a pair $(a,A)$ is in a multirelation $R$ and a pair
$(a,B)$ in a multirelation $S$, then the parallel composition of $R$
and $S$ contains the pair $(a,A\cup B)$. Starting from input $a$, the
multirelations $R$ and $S$ therefore produce the collective output
$A\cup B$ when executed in parallel. In contrast to Peleg, Parikh also
imposes additional conditions on multirelations. In particular, they
must be up-closed: $(a,A)\in R$ and $A\subseteq B$ imply $(a,B)\in R$.

In $\CDL$, modal box and diamond operators are associated with the
multirelational semantics as they are associated with a relational
semantics in $\PDL$. An expression $[\alpha]\varphi$ means that after
every terminating execution of program $\alpha$, property $\varphi$
holds, whereas $\langle\alpha\rangle\varphi$ means that there is a
terminating execution of $\alpha$ after which $\varphi$ holds. In
$\CDL$, as in $\PDL$, boxes and diamonds are related by De Morgan
duality: $[\alpha]\varphi$ holds if and only if
$\neg\langle\alpha\rangle\neg\varphi$ holds. The axioms of $\CDL$
describe how the programming constructs of external choice, sequential
and parallel composition, and (sequential) iteration interact with the
modalities. $\CDL$ can as well be seen as a generalisation of
dual-free $\GL$.

Wijesekera and Nerode~\cite{NerodeWijesekera90,WijesekeraNerode05} as
well as Goldblatt~\cite{Goldblatt92} have generalised $\CDL$ to
situations where boxes and diamonds are no longer dual. $\GL$ has been
applied widely in game and social choice theory. A bridge between the
two formalisms has recently been built by van Benthem et
al.~\cite{BenthemGL08} to model simultaneous games as they arise in
algorithmic game theory. Peleg has added notions of synchronisation
and communication to $\CDL$~\cite{Peleg87a}.  Parikh's semantics of
up-closed multirelations and its duality to monotone predicate
transformers has reappeared in Back and von Wright's refinement
calculus~\cite{Back98} and the approach to multirelational semantics
of Rewitzky and
coworkers~\cite{Rewitzky03,RewitzkyB06,MartinCR07}. Up-closed
multirelations have also been studied more abstractly as a variant of
Kleene algebra~\cite{FurusawaNT09,NishizawaTF09}. Finally, the
transitions in alternating automata can be represented as
multirelations.

This suggests that $\CDL$ and its variants are relevant to games and
concurrency; they provide insights in games for concurrency and for
concurrency in games. Despite this, beyond the up-closed case, the
algebra of multirelations, as a generalisation of Kleene
algebras~\cite{Kozen94} and Tarski's relation algebra
(cf.~\cite{Maddux}), has never been studied in detail and
\emph{concurrent dynamic algebras} as algebraic companions of $\CDL$
remain to be established. This is in contrast to $\PDL$ where the
corresponding dynamic algebras~\cite{Pratt80} and test
algebras~\cite{Nemeti81,TrnkovaR87,Pratt91} are well studied.

An algebraic reconstruction of $\CDL$ complements the logical one in
important ways. Algebras of multirelations yield abstract yet
fine-grained views on the structure of simultaneous games; they might
also serve as intermediate semantics for shared-variable concurrency,
where interferences have been resolved. The study of dynamic and test
algebras shows how modal algebras arise from Kleene and relation
algebras in particularly simple and direct ways, and powerful tools
from universal algebra and category theory are available for their
analysis. Reasoning with modal algebras is essentially first-order
equational and therefore highly suitable for mechanisation and
automation. In the context of $\CDL$ this would make the design of
tools for analysing games or concurrent programs particularly simple
and flexible.

Our main contribution is an axiomatisation of concurrent dynamic
algebras. It is obtained from axiomatisations of the algebra of
multirelations which generalise modal Kleene
algebras~\cite{DesharnaisMS06,DesharnaisStruth11,DesharnaisStruth08}. In
more detail, our main results are as follows.
\begin{itemize}
\item We investigate the basic algebraic properties of the
  multirelational semantics of $\CDL$. It turns out that those of
  sequential composition are rather weak---the operation is, for
  instance, non-associative---while concurrent composition and union
  form a commutative idempotent semiring. We also find a new
  interaction law between sequential and concurrent composition. In
  addition we investigate special properties of subidentities, which
  serve as propositions and tests in $\CDL$, and of multirelational
  domain and antidomain (domain complement) operations.
\item We axiomatise variants of semirings (called \emph{proto-dioids}
  and \emph{proto-trioids}) which capture the basic algebra of
  multirelations without and with concurrent composition. We expand
  these structures by axioms for domain and antidomain operations,
  explore the algebraic laws governing these operations and
  characterise the subalgebras of domain elements, which serve as
  state or proposition spaces in this setting. We also prove soundness
  with respect to the underlying multirelational model.
\item We define algebraic diamond and box operators from the domain
  and antidomain ones as abstract preimage operators and their De
  Morgan duals and show that algebraic counterparts of the axioms of
  star-free $\CDL$ can be derived in this setting. The diamond axioms
  of $\CDL$ are obtained over a state space which forms a distributive
  lattice; the additional box axioms are derivable over a boolean
  algebra.
\item We investigate the Kleene star (or reflexive transitive closure
  operation) in the multirelational model and turn the resulting laws
  into axioms of \emph{proto-Kleene algebras with domain} and
  \emph{antidomain} as well as \emph{proto-bi-Kleene algebras with
    domain} and \emph{antidomain}. The latter two allow us to derive
  the full set of $\CDL$ axioms; they are therefore informally called
  \emph{concurrent dynamic algebras}. Once more we prove soundness
  with respect to the underlying multirelational model.
\item Finally, we study notions of finite
iteration for the Kleene star in the multirelational setting and
refute the validity of a variant of Segerberg's axiom of $\PDL$.
\end{itemize}
The complete list of concurrent dynamic algebra axioms can be found in
Appendix 1.

Our analysis of the multirelational model and our axiomatisations are
minimalistic in the sense that we have tried to elaborate the most
general algebraic conditions for deriving the $\CDL$ axioms. Many
interesting properties of that model have therefore been ignored. Due
to the absence of associativity of sequential composition and of left
distributivity of sequential composition over union, many proofs seem
rather fragile and depend on stronger algebraic properties of special
elements. Sequential composition is, for instance, associative if one
of the particpating multirelations is a domain or antidomain
element. This requires a significant generalisation of previous
approaches to Kleene algebras with domain and
antidomain~\cite{DesharnaisStruth11,DesharnaisStruth08}.

Moreover, proofs about multirelations are rather tedious due to the
complexity of sequential composition---specifying the family of sets
$C_b$ requires second-order quantification. We have therefore
formalised and verified the most important proofs with the Isabelle
proof assistant~\cite{NipkowPW02} (see Appendix~3 for a list). Thus
our work is also an exercise in formalised mathematics. The complete
code can be found
online\footnote{\url{http://www.dcs.shef.ac.uk/~georg/isa/cda}}. We
also present all manual proofs in order to make this article
selfcontained; the less interesting ones have been delegated to
Appendix~2.

%%%%%%%%%%%%%%%%%%%%%%%%%%%%%%%%%%%%%%%%%%%%%%%%%%%%%%%%%%%%%%%%%%%%%%%%%%%%%

\section{Multirelations}\label{S:multirelations}

A \emph{multirelation} $R$ over a set $X$ is a subset of $X\times
2^X$.  Inputs $a\in X$ are related by $R$ to outputs $A\subseteq X$;
each single input $a$ may be related to many subsets of
$X$. The set of all multirelations over $X$ is denoted $M(X)$.

An intuitive interpretation is the accessibility or reachability in a
(directed) graph: $(a,A)$ means that the set $A$ of vertices is
reachable from vertex $a$ in the graph. $(a,\emptyset)$ means that no
set of vertices is reachable from $a$, which makes $a$ a terminal
node. This is different from $(a,A)$ not being an element of a
multirelation for all $A\subseteq X$.

By definition, $(a,A)$ and $(a,\emptyset)$ can be elements of the same
multirelation. This can be interpreted as a system, program or player
making an ``interal'', existential or angelic choice to access either
$A$ or $\emptyset$. The elements of $A$ can therefore be seen as
``external'', universal or demonic choices made by an environment,
scheduler or adversary player.

This ability to capture internal and external choices makes
multirelations relevant to games and game logics~\cite{Parikh85},
demonic/angelic semantics of programs~\cite{Back98,MartinCR07},
alternating automata and concurrency~\cite{Peleg87}. Different
applications, however, require different definitions of operations on
multirelations.  The one used in the concurrent setting by
Peleg~\cite{Peleg87} and Goldblatt~\cite{Goldblatt92} is the most
general one and we follow it in this article.

\begin{example}
  Let $X=\{a,b,c,d\}$. Then
  \begin{equation*}
    R=\{(a,\emptyset),(a,\{d\}),(b,\{a\}),(b,\{b\}),(b,\{a,b\})(c,\{a\}),(c,\{d\})\}
  \end{equation*}
  is a multirelation over $X$. Vertex $a$ can alternatively reach no
  vertex at all---the empty set---or the singleton set $\{d\}$. Vertex
  $b$ can either reach set $\{a\}$, set $\{b\}$ or their union
  $\{a,b\}$. Vertex $c$ can either reach set $\{a\}$ or set $\{d\}$,
  but not their union. Vertex $d$ cannot even reach the empty set; no
  execution from it is enabled. This is in contrast to the situation
  $(a,\emptyset)$, where execution is enabled from $a$, but no state
  can be reached.\qed
\end{example}

Peleg defines the following operations of sequential and concurrent
composition of multirelations. Let $R$ and $S$ be multirelations over
$X$. The \emph{sequential composition} of $R$ and $S$ is the
multirelation
\begin{equation*}
  R\cdot S = \{(a,A) \mid \exists B.\ (a,B)\in R \wedge \exists f.\ (\forall b\in B.\ (b,f(b))\in S) \wedge A=\bigcup f(B)\}.
\end{equation*}
The \emph{unit of sequential composition} is the multirelation
\begin{equation*}
  1_\sigma = \{(a,\{a\})\mid a\in X\}.
\end{equation*}
The \emph{parallel composition} of $R$ and $S$ is the multirelation
\begin{equation*}
  R\|S = \{(a,A\cup B) \mid (a,A)\in R\wedge (a,B)\in S\}.
\end{equation*}
The \emph{unit of parallel composition} is the multirelation
\begin{equation*}
  1_\pi= \{(a,\emptyset)\mid a\in X\}.
\end{equation*}
The \emph{universal multirelation} over $X$ is
 \begin{equation*}
 U =\{(a,A)\mid a \in X \wedge A \subseteq X\}. 
 \end{equation*}

 In the definition of sequential composition, $f(B) = \{f(b)\mid b \in
 B\}$ is the image of $B$ under $f$. The intended meaning of $(a,A)
 \in R\cdot S$ is as follows: the set $A$ is reachable from vertex $a$
 by $R\cdot S$ if some intermediate set $B$ is reachable from $a$ by
 $R$, and from each vertex $b\in B$ a set $A_b$ is reachable
 (represented by $f(b)$) such that $A = \bigcup_{b \in B}A_b=\bigcup
 f(B)$.  Thus, from each vertex $b\in B$, the locally reachable set
 $f(b)$ contributes to the global reachability of $A$. We write
 $G_f(b)=(b,f(b))$ for the graph of $f$ at point $b$, and
 $G_f(B)=\{G_f(b)\mid b\in B\}$ for the graph of $f$ on the set
 $B$. We can then write
\begin{equation*}
  (a,A) \in R \cdot S \Leftrightarrow \exists B.\ (a,B) \in R \wedge \exists f.\ G_f(B)\subseteq S \wedge A = \bigcup f(B).
\end{equation*}
This definition of sequential composition is subtly different to the
one used by Parikh~\cite{Parikh85} in game logics, which appears also
in papers on multirelational semantics and monotone predicate
transformers. In addition, Parikh considers up-closed
multirelations. This leads not only to much simpler proofs, but also
to structural differences. Peleg has argued that up-closure is not
desirable for concurrency since it makes all programs---even
tests---automatically nondeterministic.

The sequential identity $1_\sigma$ is defined similarly to the identity
relation or identity function. It is given by (the graph of) the
embedding $\lambda x.\{x\}$ into singleton sets.

In a parallel composition, $(a,A) \in R\|S$ if $A$ is reachable from
$a$ by $R$ or $S$ in collaboration, that is, each of $R$ and $S$ must
contribute a part of the reachability to $A$.

The parallel identity $1_\pi$ is the function $\lambda x.\
\emptyset$, which does not reach any set from any vertex. Two
interpretations of a pair $(a,\emptyset)$ suggest themselves: it might
be the case that nothing is reachable from $a$ due to an error or due
to nontermination.

\begin{example}
Consider the multirelations 
\begin{equation*}
  R=\{(a,\{b,c\})\},\qquad S=\{(b,\{b\})\}, \qquad T=\{(b,\{b\}),(c,\emptyset)\}.
\end{equation*}
Then $R\cdot S = \emptyset$ because $S$ cannot contribute from
$c$. Moreover, $R\cdot T= \{(a,\{b\})\}$. Finally, $T\cdot S = T$,
since, from $c$, the empty set is the only intermediate set which
satisfies the conditions for $S$ and $A$ above.\qed
\end{example}
\begin{example}
  Consider the multirelations 
  \begin{equation*}
    R=\{(a,\{a,b\})\},\qquad S=\{(a,\{b,c\}),(b,\{b\})\},\qquad T=\{(b,\emptyset)\}.
  \end{equation*}
Then $R||S = \{(a,\{a,b,c\})\}$ and $S||T = \{(b,\{b\})\}$.\qed
\end{example}

%%%%%%%%%%%%%%%%%%%%%%%%%%%%%%%%%%%%%%%%%%%%%%%%%%%%%%%%%%%%%%%%%%%%%%%%%%

\section{Basic Laws for Multirelations}\label{S:mulproperties}
 
% All theorems in this section with the exception of
% Lemma~\ref{P:seqlaws} (3) have been verified with Isabelle. We present
% their proofs only for the sake of completeness.

The definition of sequential composition is higher-order and some of
our proofs of our proofs use higher-order Skolemisation, which is an
instance of the Axiom of Choice:
\begin{equation*}
  (\forall a\in A.\exists b.\ P(a,b)) \Leftrightarrow (\exists f.\forall a\in A.\ P(a,f(a))).
\end{equation*}

First we derive some basic laws of sequential composition.
\begin{lemma}\label{P:seqlaws}
  Let $R$, $S$ and $T$ be multirelations.
  \begin{enumerate}
  \item $R\cdot 1_\sigma=R$ and $1_\sigma\cdot R = R$,
\item $\emptyset\cdot R= \emptyset$,
\item $(R \cdot S) \cdot T \subseteq R \cdot (S\cdot T)$,
  \item $(R\cup S)\cdot T = R\cdot T\cup S\cdot T$,
%\item $(\bigcup_{i\in I}R_i)\cdot S = \bigcup_{i\in I}(R_i\cdot S)$.
   \item $R\cdot S\cup R\cdot T \subseteq R\cdot (S\cup T)$.
\end{enumerate}
\end{lemma}
See Appendix~2 for proofs.  Property (1) confirms that $1_\sigma$ is
indeed an identity of sequential composition. Property (2) shows that
$\emptyset$ is a left annihilator. It is, however, not a right
annihilator by Lemma~\ref{P:counterexamples} below. Similarly, (3) is
a weak associativity law which, again by
Lemma~\ref{P:counterexamples}, cannot be strengthened to an
identity. In fact, (3) is not needed for the algebraic development in
this paper; it is listed for the sake of completeness. Property (5) is
a left subdistributivity law for sequential composition, which, again
by Lemma~\ref{P:counterexamples}, cannot be strengthened to an
identity. Left subdistributivity and right distributivity imply that
sequential composition is left and right isotone:
\begin{equation*}
  R\subseteq S\Rightarrow T\cdot R\subseteq T\cdot S,\qquad R\subseteq S\Rightarrow R\cdot T\subseteq S\cdot T.
\end{equation*}

Next we verify some basic laws of concurrent composition. These reveal
more pleasant algebraic structure.
\begin{lemma}\label{P:conclaws}
Let $R$, $S$ and $T$ be multirelations.
  \begin{enumerate}
  \item $(R\|S)\|T = R\|(S\|T)$,
  \item $R\|S = S\|R$,
  \item $R\|1_\pi= R$,
\item $R\|\emptyset = \emptyset$,
  \item $R\|(S\cup T) = R\|S \cup R\|T$.
  \end{enumerate}
\end{lemma}
See Appendix~2 for proofs. This show that multirelations under union
and parallel composition form a commutative dioid, as introduced in
Section~\ref{S:axioms}.  It follows that concurrent composition is
left and right isotone:
\begin{equation*}
  R\subseteq S \Rightarrow T\|R\subseteq T\|S,\qquad R\subseteq S \Rightarrow R\|T\subseteq S\|T.
\end{equation*}

Next we establish an important interaction law between sequential and
concurrent composition; a right subdistributivity law of sequential
over concurrent composition.
  \begin{lemma}\label{P:interaction}
    Let $R$, $S$ and $T$ be multirelations. Then
    \begin{equation*}
    (R\|S)\cdot T \subseteq (R\cdot T)\|(S\cdot T).    
    \end{equation*}
  \end{lemma}
  See Appendix~2 for a proof. Once more, this general law is not
  needed for our algebraic development. We use a full right
  distributivity law that holds in particular cases.

  Finally, counterexamples show that the algebraic properties studied
  so far are sharp.

%   These have been obtained with Isabelle's counterexample generators
%   Nitpick or Quickcheck.

\begin{lemma}\label{P:counterexamples}
  There are multirelations $R$, $S$ and $T$ such that
\begin{enumerate}
\item $R\cdot \emptyset \neq \emptyset$,
\item $R \cdot (S\cdot T) \not\subseteq (R \cdot S) \cdot T$,
\item $R\cdot (S\cup T) \not\subseteq R\cdot S\cup R\cdot T$,
\item $(R\cdot T)\|(S\cdot T)\not\subseteq (R\|S)\cdot T$.
\end{enumerate}
\end{lemma}
\begin{proof}
  \begin{enumerate}
\item Let $R=\{(a,\emptyset)\}$. Then
$    (a,A) \in R\cdot S \Leftrightarrow \exists f.\ G_f(B)\in S\wedge A = \bigcup f(\emptyset)
 \Leftrightarrow A = \emptyset$.
Hence, in this particular case, $R\cdot \emptyset = \{(a,\emptyset)\}\neq \emptyset$.
\item 
Let $R=\{(a,\{a,b\}),(a,\{a\}),(b,\{a\})\}$ and $S=\{(a,\{a\}),(a,\{b\})$. Then
\begin{align*}
  (R\cdot R)\cdot S &= \{(a,\{a\}),(a,\{b\}),(b,\{a\}),(b,\{b\})\}\\ 
&\subset \{(a,\{a,b\}),((a,\{a\}),(a,\{b\}),(b,\{a\}),(b,\{b\})\}\\
&= R\cdot (R\cdot S).
\end{align*}
\item Consider $ R=\{(a,\{a,b\})\}$, $S=\{(a,\{a\})\}$, and
  $T=\{(b,\{b\})\}$.  It follows that $S\cup T=
  \{(a,\{a\}),(b,\{b\})\}$ and $R\cdot (S\cup T)=R$, but $R\cdot S=
  R\cdot T=\emptyset$, whence $R\cdot S\cup R\cdot T=\emptyset$.
\item Let $R=\{(a,\{a\})\}$ and $S=\{(a,\{a\}),(a,\{b\})\}$. Then 
  \begin{equation*}
    (R\|R)\cdot T = T \subset \{(a,\{a\}),(a,\{b\}),(a,\{a,b\})\}= (R\cdot T)\|(R\cdot T).
  \end{equation*}
  \end{enumerate}
\end{proof}

The following Hasse diagrams are useful for visualising multirelations
and finding counterexamples. We depict the multirelations $R$ and $S$
from case (2) in the Hasse diagram of the carrier set in
Figure~\ref{Fig:one}.
\begin{figure}[h]
  \centering

\begin{equation*}
  \def\labelstyle{\normalsize}
\xymatrix@C=5pt{
&ab\ar@/_/[dl]\ar@(ul,ur)&\\
a\ar@/_/[ur]\ar@(dl,ul)&&b\ar[ll]\\
}
\qquad\qquad\qquad
\xymatrix{
&\\
a\ar[r]\ar@(dl,ul)&b
}
\end{equation*}
\caption{Diagrams for $R$ and $S$ in the proof of Lemma~\ref{P:counterexamples}(2)}
\label{Fig:one}
\end{figure}
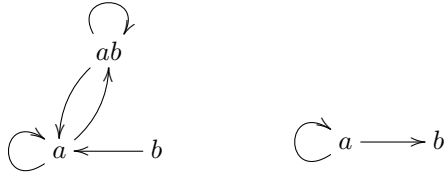
We write $ab$ as shorthand for the set $\{a,b\}$. The arrows $a\to a$,
$a\to ab$ and $b\to a$ correspond to the pairs in $R$. The ``virtual''
arrows $ab\to ab$ and $ab\to a$ have been added to indicate which
states are reachable from the set $ab$ by $R$. We have omitted the
empty set because it is not reachable. 

The resulting lifting of the multirelation of type $X\to 2^X$ to a
relation $2^X\times 2^X$ allows us to compute powers of $R$ and
products such as $R\cdot S$ by using relational composition, that is,
by chasing reachability arrows directly in the diagram. It is
reminiscent of Rabin and Scott's construction of deterministic finite
automata from nondeterministic ones. A systematic study of this
lifting will be the subject of another article.

Accordingly, we compute $R\cdot R$, $R\cdot S$, $(R\cdot R)\cdot S$
and $R\cdot (R\cdot S)$ as depicted in Figure~\ref{Fig:two}.
\begin{figure}[h]
  \centering
\begin{gather*}
  \def\labelstyle{\normalsize}
\xymatrix@C=5pt{
&ab\ar@/_/[dl]\ar@(ul,ur)&\\
a\ar@/_/[ur]\ar@(dl,ul)&&b\ar[ll]\ar[lu]
}
\qquad\qquad
\xymatrix{
&\\
a\ar@/_/[r]\ar@(dl,ul)&b\ar@/_/[l]\ar@(ur,dr)
}
%\qquad\qquad\qquad
\\
\\
\\
\xymatrix{
&\\
a\ar@/_/[r]\ar@(dl,ul)&b\ar@/_/[l]\ar@(ur,dr)
}
\qquad\qquad\qquad
\xymatrix@C=5pt{
&ab\ar@/_/[dl]\ar@(ul,ur)\ar[dr]&\\
a\ar@/_/[ur]\ar@(dl,ul)\ar@/_/[rr]&&b\ar@/_/[ll]\ar@(ur,dr)
}
\end{gather*}
  
  \caption{Diagrams for $R\cdot R$, $R\cdot S$, $(R\cdot R)\cdot S$ and $R\cdot (R\cdot S)$ with $R$ and $S$ from the proof of Lemma~\ref{P:counterexamples}(2)}
  \label{Fig:two}
\end{figure}
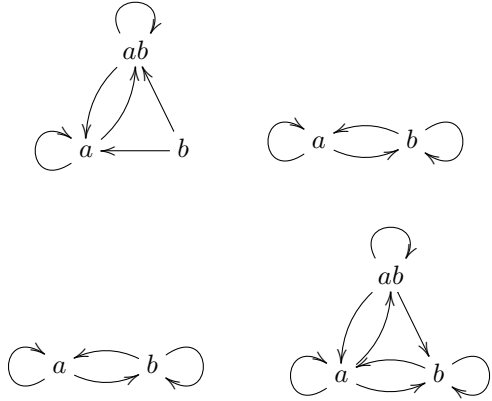

We even have a counterexample to $R\cdot (R\cdot R) \subseteq (R\cdot R)\cdot
R$. Consider the multirelation $R= \{(a,\{c\}),(b,\{a,c\}),(c,\{b\}),(c,\{c\})\}$. Then 
\begin{equation*}
  R\cdot R = \{(a,\{b\}),(a,\{c\}),(b,\{c\}),(b,\{b,c\}),(c,\{b\}),(c,\{c\}), (c,\{a,c\})\}.
\end{equation*}
Therefore
\begin{align*}
  R\cdot (R\cdot R) =& \{(a,\{b\}),(a,\{c\}),(a,\{a,c\}),\\
&\ \ (b,\{b\}),(b,\{c\}),(b,\{a,c\}),(b,\{b,c\}),(b,\{a,b,c\}),\\
&\ \ (c,\{b\}),(c,\{c\}),(c,\{a,c\}),(c,\{b,c\})\}\\
\supset& \{(a,\{b\}),(a,\{c\}),(a,\{a,c\}),\\
&\ \ (b,\{b\}),(b,\{c\}),(b,\{a,c\}),(b,\{a,b,c\}),\\
&\ \ (c,\{b\}),(c,\{c\}),(c,\{b,c\})\}\\
=& (R\cdot R)\cdot R.
\end{align*}
This hints at complications in the definition of finite iteration of
multirelations, which is considered in
Section~\ref{S:finiteiteration}.

%%%%%%%%%%%%%%%%%%%%%%%%%%%%%%%%%%%%%%%%%%%%%%%%%%%%%%%%%%%%%%%%%%%%%%%%%%%%

\section{Stronger Laws for Sequential Subidentities}\label{S:subidlaws}

% Again, all properties in this section have been verified with
% Isabelle.

A multirelation $P$ is a \emph{(sequential) subidentity} if
$P\subseteq 1_\sigma$. As mentioned in Section~\ref{S:multirelations},
$1_\sigma=\lambda x.\{x\}$ embeds $X$ into $2^X$. Every sequential
subidentity is therefore a partial embedding. We usually write $P$ or
$Q$ for subidentities. We write $\iota=\lambda x.\{x\}$ for the
embedding of $X$ into $2^X$. One can see $G_\iota(a)$ also as a
lifting of a point $a\in X$ to a multirelational ``point'' $(a,\{a\})$
and $G_\iota(A)$ as a lifting of a set $A$ to a subidentity.

More generally, this yields an isomorphism between points and
multirelational points as well as sets and subidentities.

The next lemma shows that multiplying a multirelation with a
subidentity from the left or right amounts to an input or output
restriction. 
 \begin{lemma}\label{P:subidinout}
   Let $R$ be a multirelation and $P$ a subidentity.
   \begin{enumerate}
   \item $(a,A) \in R\cdot P \Leftrightarrow (a,A)\in R\wedge
     G_\iota(A)\subseteq P$,
   \item $(a,A) \in P\cdot R \Leftrightarrow G_\iota(a) \in P \wedge
     (a,A) \in R$.
   \end{enumerate}
 \end{lemma}
 See Appendix~2 for proofs. These properties help us to verify that
 subidentities satisfy equational associativity and interaction laws
 as well as a left distributivity law.
 \begin{lemma}\label{P:subidlaws}
   Let $R$, $S$ and $T$ be multirelations.
\begin{enumerate}
\item $(R\cdot S)\cdot T = R \cdot (S\cdot T)$ if $R$, $S$ or $T$ is a
  subidentity,
\item  $(R\|S)\cdot T = (R\cdot T)\|(S\cdot T)$ if $T$ is a subidentity,
\item $R \cdot (S\cup T) = R\cdot S \cup R\cdot T$ if $R$ is a
  subidentity.
 \end{enumerate}  
\end{lemma}
See Appendix~2 for proofs. Lemma~\ref{P:subidlaws} is essential for
deriving the axioms of concurrent dynamic algebra.

In addition, it is straightforward to verify that the sequential
subidentities form a boolean subalgebra of the algebra of
multirelations over $X$. The empty set is the least element of this
algebra and $1_\sigma$ its greatest element. Join is union and meet
coincides with sequential composition, which is equal to parallel
composition in this special case. The boolean complement of a
subidentity $\bigcup_{a\in A}\{G_\iota(a)\}$, for some set $A\subseteq
X$, is the subidentity $\bigcup_{b\in X - A}\{G_\iota(b)\}$.

Subidentities play an important role in providing the state spaces of
modal operators in concurrent dynamic algebras. In our axiomatisation,
however, they arise only indirectly through definitions of domain and
antidomain elements. In the concrete case of multirelations these are
described in the next section.

%%%%%%%%%%%%%%%%%%%%%%%%%%%%%%%%%%%%%%%%%%%%%%%%%%%%%%%%%%%%%%%%%%%%%%%%%%%%%

\section{Domain and Antidomain of  Multirelations}\label{S:multireldom}

%Again, all properties of this section have been verified with Isabelle.

This section presents the second important step towards concurrent
dynamic algebra within the multirelational model: the definitions of
domain and antidomain operations and the verification of some of their
basic properties. These are then abstracted into algebraic domain and
antidomain axioms, which, in turn, allow us to define the modal box
and diamond operations of concurrent dynamic algebra.

The \emph{domain} of a multirelation $R$ is the multirelation
 \begin{equation*}
   d(R)=\{G_\iota(a)\ |\ \exists A.\ (a,A)\in R\}.
 \end{equation*}
The \emph{antidomain} of a multirelation $R$ is the multirelation
\begin{equation*}
  a(R) = \{G_\iota(a)\}\ | \ \neg\exists A.\ (a,A)\in R\}.
\end{equation*}
Domain and antidomain elements are therefore boolean complements of
each other.

The next lemmas collect some of their basic properties which justify
the algebraic axioms in Section~\ref{S:axioms}.

 \begin{lemma}\label{P:domprops}
Let $R$ and $S$ be multirelations.
   \begin{enumerate}
   \item $d(R)\subseteq 1_\sigma$,
    \item $d(R)\cdot R= R$,
   \item $d(R\cup S)=d(R)\cup d(S)$,
   \item $d(\emptyset)=\emptyset$,
   \item $d(R\cdot S)=d(R\cdot d(S))$,
   \item $d(R\|S)=d(R)\cap d(S)$,
   \item $d(R)\|d(S) = d(R)\cdot d(S)$.
   \end{enumerate}
 \end{lemma}
 See Appendix~2 for proofs. Most of these laws are similar to those of
 relational domain, but properties (6) and (7) are particular to
 multirelations. Property (1) shows that domain elements are
 subidentities. According to (2), a multirelation is preserved by
 multiplying it from the left with its domain element. According to
 (3) and (4), domain is strict and additive: the domain of the union
 of two multirelations is the union of their domains and the domain of
 the empty set is the empty set. The locality property (5) states that
 it suffices to know the domain of the second multirelation when
 computing the domain of the sequential composition of two
 multirelations. By (6), the domain of a parallel composition of two
 multirelations is the intersection of their domains. Finally, by (7),
 the parallel composition of two domain elements equals their
 interesection. More generally, parallel composition of sequential
 subidentities is meet.

 An intuitive explanation of domain is that it yields the set of all
 states from which a multirelation is enabled. Accordingly, by (3),
 the union of two multirelations is enabled if one of them is enabled,
 whereas, by (6), their parallel composition is enabled if both are
 enabled. It follows immediately from the definition that
 $d(\{(a,\emptyset)\})=\{(a,\{a\})\}$. Hence the multirelation
 $\{(a,\emptyset)\}$ is enabled, but does not yield an output.

 The next lemma, proved in Appendix~2, links domain and
 antidomain. It shows, in particular, that domain and antidomain
 elements are complemented.

\begin{lemma}\label{P:antidomprops1}
Let $R$ be a multirelation.
  \begin{enumerate}
  \item $a(R)= 1_\sigma\cap -d(R)$,
  \item $d(R)=a(a(R))$,
  \item $d(a(R))=a(R)$.
  \end{enumerate}
\end{lemma}
Many essential properties of antidomain can now be derived by De
Morgan duality.

\begin{lemma}\label{P:antidomprops2}
Let $R$ and $S$ be multirelations.
  \begin{enumerate}
  \item $a(R)\cdot R = \emptyset$,
  \item $a(R\cdot S) = a(R\cdot d(S))$,
  \item $a(R)\cup d(R) = 1_\sigma$,
  \item $a(R\cup S) = a(R)\cdot a(S)$,
  \item $a(R\|S) = a(R)\cup a(S)$,
  \item $a(R)\|a(S) = a(R)\cdot a(S)$.
  \end{enumerate}
\end{lemma}
See Appendix~\ref{A:proofs} for proofs. If $d(R)$ describes those
states from which mutirelation $R$ is enabled, then $a(R)$ models
those where $R$ is not enabled. Property (1) says that antidomain
elements are left annihilators: $R$ cannot be executed from states
where it is not enabled. Property (2) is a locality property similar
to that in Lemma~\ref{P:domprops}(5). Property (3) is a
complementation law between domain and antidomain elements. It implies
that antidomain elements are sequential subidentities. Properties (4)
to (6) are the obvious De Morgan duals of domain properties.

Finally, and crucially for our purposes, domain and antidomain
elements support stronger associativity and distributivity properties.

\begin{corollary}\label{P:domassocinter}
Let $R$, $S$ and $T$ be multirelations.
  \begin{enumerate}
  \item $(R\cdot S)\cdot T = R\cdot (S\cdot T)$ if $R$, $S$ or $T$ is
    a domain or antidomain element,
  \item $(R\|S)\cdot T = (R\cdot T)\|(S\cdot T)$, if $T$ is a
    domain or antidomain element,
  \item $R\cdot (S\cup T) = R\cdot S\cup R\cdot T$ if $R$ is a domain
    or antidomain element.
  \end{enumerate}
\end{corollary}
\begin{proof}
  By Lemma~\ref{P:domprops}(1) and~\ref{P:domprops}(3), domain and antidomain elements are subidentities. The
  results then follow by Lemma~\ref{P:subidlaws}.
\end{proof}

Domain and antidomain satisfy, of course, additional properties. We
have only presented those needed to justify the abstract domain and
antidomain axioms in the following section. Further ones can then be
derived by simple equational reasoning at the abstract level from
those axioms; a considerable simplification.

%%%%%%%%%%%%%%%%%%%%%%%%%%%%%%%%%%%%%%%%%%%%%%%%%%%%%%%%%%%%%%%%%%%%%%%%%%%%%%

\section{Axioms for Multirelations with Domain and Antidomain}\label{S:axioms}

We have now collected sufficiently many facts about multirelations to
abstract the domain and antidomain laws from the previous section into
algebraic axioms. The approach is inspired by the axiomatisation of
domain semirings~\cite{DesharnaisStruth11} in the relational setting
and the weakening of these axioms to families of
near-semirings~\cite{DesharnaisStruth08}. In those approaches,
however, sequential composition is associative, which considerably
simplifies proofs and leads to simpler axiomatisations. Here we can
only assume associativity, interaction and left distributivity in the
presence of domain and antidomain elements, which holds in the
multirelational model according to Corollary~\ref{P:domassocinter} and
yields just the right assumptions for reconstructing concurrent
dynamic logic.

We keep the development modular so that it captures also
multirelational semirings and Kleene algebras without concurrent
composition. We expect that the axioms of Parikh's game logic can be
derived from that basis.

A \emph{proto-dioid} is a structure $(S,+,\cdot,0,1)$ such that
$(S,+,0)$ is a semilattice with least element $0$ and the following
additional axioms hold:
   \begin{gather*}
 %  (x\cdot y)\cdot z \le x\cdot (y\cdot z),\qquad 
1\cdot x = x,\qquad x\cdot 1=x,\\
x\cdot y+x\cdot z \le x\cdot (y+z),\qquad (x+ y)\cdot z = x\cdot z+y\cdot z,\qquad 0\cdot x=0.
  \end{gather*}
  Here, $\le$ is the semilattice order defined, as usual, by $x\le
  y\Leftrightarrow x+y=y$.

  We do not include the weak associativity law $(x\cdot y)\cdot z \le
  x\cdot (y\cdot z)$, although it is present in multirelations
  (Lemma~\ref{P:seqlaws}(3)). It is independent from our axioms.

  A \emph{dioid} is a proto-dioid in which multiplication is
  associative for all elements and the left distributivity law $x\cdot
  (y+z) = x\cdot y+x\cdot z$ and the right annihilation law $x\cdot
  0=0$ hold. A dioid is \emph{commutative} if multiplication is commutative:
  $x\cdot y=y\cdot x$.

  A \emph{proto-trioid} is a structure
  $(S,+,\cdot,\|,0,1_\sigma,1_\pi)$ such that $(S,+,\cdot,0,1_\sigma)$
  is a proto-dioid and $(S,+,\|,0,1_\pi)$ is a commutative dioid.

In every proto-dioid, multiplication is left-isotone, $x\le
y\Rightarrow z\cdot x\le z\cdot y$.

A \emph{domain proto-dioid} (\emph{dp-dioid}) is a proto-dioid
expanded by a domain operation which satisfies the \emph{domain
  associativity} axiom
  \begin{equation*}
    x\cdot (y\cdot z) = (x\cdot y) \cdot z,
  \end{equation*}
if one of $x$, $y$ or $z$ is equal to $d(w)$ for some $w$, and the domain axioms
  \begin{gather*}
    x\le d(x)\cdot x,\qquad d(x\cdot y) = d(x\cdot d(y)), d(x+y) = d(x)+d(y),\\ d(x)\le 1_\sigma,\qquad
d(0)=0.
  \end{gather*}
  The first domain axiom is called \emph{left preservation} axiom, the
  second one \emph{locality} axiom, the third one \emph{additivity}
  axiom, the fourth one \emph{subidentity} axiom and the fifth one
  \emph{strictness} axiom.

  A \emph{domain proto-trioid} (\emph{dp-triod}) is a dp-dioid which
  is also a proto-trioid and satisfies the \emph{domain interaction}
  axiom and the \emph{domain concurrency} axioms
  \begin{gather*}
     (x\|y)\cdot d(z) = (x\cdot d(z))\|(y\cdot d(z)),\qquad
  d(x\|y) = d(x)\cdot d(y),\\
 d(x)\|d(y) = d(x)\cdot d(y).
\end{gather*}

In the presence of antidomain the axioms can be simplified further.
An \emph{antidomain proto-dioid} (\emph{ap-dioid}) is a proto-dioid
expanded by an antidomain operation which satisfies the
\emph{antidomain associativity} axiom
  \begin{equation*}
    x\cdot (y\cdot z) = (x\cdot y) \cdot z,
  \end{equation*}
  where $x$, $y$ or $z$ is equal to $a(w)$ for some $w$, and satisfies
  the antidomain axioms
  \begin{gather*}
    a(x)\cdot x= 0,\qquad a(x\cdot y)=a(x\cdot a(a(y))),\qquad a(x)+a(a(x))= 1_\sigma,\\ a(x)\cdot (y+z) =a(x)\cdot y+a(x)\cdot z.
  \end{gather*}
  The first antidomain axiom is called \emph{left annihilation} axiom,
  the second one \emph{locality} axiom, the third one
  \emph{complementation} axiom and the fourth one \emph{antidomain
    left distributivity} axiom.

  An \emph{antidomain proto-trioid} (\emph{ap-trioid}) is an ap-dioid
  which is also a proto-trioid and satisfies the \emph{antidomain
    interaction} and \emph{antidomain concurrency} axioms
  \begin{gather*}
    (x\|y)\cdot a(z) = (x\cdot a(z))\|(y\cdot a(z)),\qquad
a(x\|y)=a(x)+a(y),\\ a(x)\|a(y)=a(x)\cdot a(y).
\end{gather*}

We have verified irredundancy of all domain and antidomain axioms with
Isabelle. The full set of axioms of dp-trioids and ap-trioids (with
additional axioms for the Kleene star) is listed in Appendix~1.

We can now relate the multirelational model set up in
Sections~\ref{S:mulproperties}-\ref{S:multireldom} with the abstract
algebraic definitions. The theorem is stated only for the smallest
axiomatic class; it then holds automatically in all superclasses.

\begin{theorem}\label{P:dmramrmodel}
Let $X$ be a set.
 \begin{enumerate}
\item   The structure $(M(X),\cup,\cdot,\|,\emptyset,1_\sigma,1_\pi,d)$ forms a dp-trioid.
\item The structure $(M(X),\cup,\cdot,\|,\emptyset,1_\sigma,1_\pi,a)$
  forms an ap-trioid.
\end{enumerate} 
\end{theorem}
\begin{proof}
  The union axioms follow from set theory.  The remaining proto-dioid
  axioms of sequential composition have been verified in
  Lemma~\ref{P:seqlaws}; the commutative dioid axioms of concurrent
  composition in Lemma~\ref{P:conclaws}; the domain and antidomain
  axioms in Lemma~\ref{P:domprops}, Lemma~\ref{P:antidomprops2} and
  Corollary~\ref{P:domassocinter}.
\end{proof}

We call the structure
$(M(X),\cup,\cdot,\|,\emptyset,1_\sigma,1_\pi,d)$ the \emph{full
  multirelational dp-trioid} and the structure
$(M(X),\cup,\cdot,\|,\emptyset,1_\sigma,1_\pi,a)$ the \emph{full
  multirelational ap-trioid} over $X$. Since dp-trioids and
ap-trioids are equational classes, they are closed under subalgebras,
products and homomorphic images. Hence in particular any subalgebra of
a full dp-trioid is a dp-trioid and any subalgebra of a full ap-trioid
is an ap-trioid.

%%%%%%%%%%%%%%%%%%%%%%%%%%%%%%%%%%%%%%%%%%%%%%%%%%%%%%%%%%%%%%%%%%%%%%

\section{Modal Operators}\label{S:modalops}

Following Desharnais and Struth~\cite{DesharnaisStruth11}, we define
modal box and diamond operators from domain and antidomain. In every
dp-dioid we define
\begin{equation*}
\langle x\rangle y = d(x\cdot
    y).
\end{equation*}
This captures the intuition behind the Kripke-style semantics of modal
logics. As explained in Section~\ref{S:subidlaws}, sequential
multiplication of a multirelation by a sequential subidentity from the
left and right forms an input or output restriction of that
multirelation. Therefore, $d(x\cdot y)=d(x\cdot d(y))$ abstractly
represents a generalised multirelational preimage of the subidentity
$d(y)$ under the element $x$. In other words, $\langle x\rangle
y=\langle x\rangle d(y)$ yields the set of all elements from which,
with $x$, one may reach a set which is a subset of $d(y)$. This can be
checked readily in the multirelational model: if $P\subseteq X\times
2^X$ is a sequential subidentity and $R\subseteq X\times 2^X$ a
multirelation, then
\begin{equation*}
  \langle R\rangle P = \{G_\iota(a)\mid \exists B.\ (a,B)\in R\wedge G_\iota(B)\subseteq P\}.
\end{equation*}
This abstractly represents the set of all states $A\subseteq X$ from
which $R$ may reach the set $B$, which is a subset of the set
represented by the multirelation $P$. In particular, if all output
sets of the multirelation are singletons, the case of a relational
preimage is recovered. The definition of multirelational diamonds
thus generalises the relational Kripke semantics in a natural
way.

In ap-dioids the situation is similar.  Boxes can now be
defined by De Morgan duality as well. In accordance with the
multirelational model (Lemma~\ref{P:antidomprops1}(2)) we show in
Section~\ref{S:amrdioids} that $d=a\circ a$. Then
\begin{equation*}
  \langle x\rangle y =d(x\cdot y)=a(a(x\cdot y)),\qquad [x]y = a(x\cdot a(y)).
\end{equation*}
Intuitively, one might expect that $[x]y=[x]d(y)$ models the set of
all states from which, whith $x$, one must reach sets of elements
which are all in $d(y)$. An analysis in the multirelational model,
however, shows a subtly different behaviour:
\begin{align*}
  [R]P &= \{G_\iota(a)\mid \neg\exists B.\ (a,B)\in R\cdot a(P)\}\\
  &= \{G_\iota(a)\mid \neg\exists B.\ (a,B)\in R\wedge G_\iota(B)\subseteq a(P)\}\\
  &= \{G_\iota(a)\mid \neg\exists B.\ (a,B)\in R\wedge G_\iota(B)\cap P=\emptyset\}\\
  &= \{G_\iota(a)\mid \forall B.\ (a,B)\in R\Rightarrow G_\iota(B)\cap P\neq\emptyset\}.
\end{align*}
The condition $(a,B)\in R\Rightarrow G_\iota(B)\cap P\neq\emptyset$,
which is enforced by De Morgan duality, is weaker than what we
described above. At least the standard relational case is contained in
this definition. Goldblatt~\cite{Goldblatt92}, following Nerode and
Wijesekera~\cite{NerodeWijesekera90}, has therefore argued for
replacing this condition by the more intuitive condition $(a,B)\in
R\Rightarrow G_\iota(B)\subseteq P$, which breaks De Morgan
duality. Here we follow Peleg's De Morgan dual definition and leave
the algebraisation of its alternative for future work.

%%%%%%%%%%%%%%%%%%%%%%%%%%%%%%%%%%%%%%%%%%%%%%%%%%%%%%%%%%%%%%%%%%%%%%%%%%%%%%%

\section{The Structure of DP-Trioids}\label{S:dmrdioids}

This section presents the basic laws of dp-dioids and
dp-trioids. Section~\ref{S:diaaxioms} shows that algebraic variants of
the axioms of concurrent dynamic logic, except the star axiom, can be
derived in this setting. The star is then treated in
Section~\ref{S:star}.

We write $d(S)$ for the image of the carrier set $S$ under the domain
operation $d$ and call this set the set of all \emph{domain
  elements}. We often write $p,q,r,\dots$ for domain elements. 

The following identity is immediate from locality and properties of
$1_\sigma$.
\begin{lemma}\label{P:domretraction}
  In every dp-dioid, operation $d$ is a retraction: $d\circ d =d$.
\end{lemma}
The next fact is a general property of retractions. Here it gives a
syntactic characterisation of domain elements as fixpoints of
$d$ (cf.~\cite{DesharnaisStruth11}).
\begin{proposition}\label{P:retractionlemma}
  If $S$ is a dp-dioid, then $x\in d(S) \Leftrightarrow d(x)=x$.
\end{proposition}
This characterisation helps checking closure properties of domain
elements. We first prove some auxiliary properties
(cf.~Appendix~2).

\begin{lemma}\label{P:protodomprops}
In every dp-dioid,
  \begin{enumerate}
\item $x\le y\Rightarrow d(x)\le d(y)$,
  \item $d(x)\cdot x = x$,
\item $d(x\cdot y)\le d(x)$,
\item $x\le 1_\sigma \Rightarrow x \le d(x)$,
\item $d(d(x)\cdot y)=d(x)\cdot d(y)$.
\end{enumerate}
\end{lemma}
The \emph{domain export} law (5) is instrumental in proving further
domain laws.

\begin{proposition}\label{P:domaindl}
  Let $S$ be a dp-dioid. Then $d(S)$ is a subalgebra of $S$
  which forms a bounded distributive lattice.
\end{proposition}
\begin{proof}
First we check that $d(S)$ is closed under the operations, using the fixpoint property $d(x)=x$ from Proposition~\ref{P:retractionlemma}.
\begin{itemize}
\item $d(0)=0$ is an axiom.
\item $d(1_\sigma) = 1_\sigma$ follows from Lemma~\ref{P:protodomprops}(1).
\item $d(d(x)+d(y))=d(x)+d(y)$ follows from additivity and idempotency
  of domain.
\item $d(d(x)\cdot d(y))=d(x)\cdot d(y)$ follows from domain export
  (Lemma~\ref{P:protodomprops}(5)) and locality.
\end{itemize}
Next we verify that the subalgebra forms a distributive lattice with
least element $0$ and greatest element $1_\sigma$.
\begin{itemize}
\item It is obvious that $1_\sigma$ is the greatest and $0$ the least
  element of $d(S)$.
\item Associativity of domain elements follows from the dp-dioid
  axioms.
\item $d(x)\cdot d(y)=d(y)\cdot d(x)$. We show that $d(x)\cdot d(y)\le
  d(y)\cdot d(x)$; the converse direction being symmetric.
  \begin{equation*}
    d(x)\cdot d(y) = d(d(x)\cdot d(y))\cdot d(x)\cdot d(y)=d(x)\cdot d(y)\cdot d(y)\cdot d(x) \le d(y)\cdot d(x),
  \end{equation*}
  using Lemma~\ref{P:protodomprops}(2), domain export, associativity
  of domain elements and the fact that domain elements are
  subidentities.
\item $d(x)\cdot d(x)=d(x)$ holds since 
\begin{equation*}
d(x)=d(d(x)\cdot x)=d(x)\cdot
  d(x)
\end{equation*}
 by Lemma~\ref{P:protodomprops}(2) and domain export.
\end{itemize}
It follows that $(d(S),\cdot 0,1)$ is a bounded meet semilattice with
meet operation $\cdot$. It is also clear that $(d(S),+,0,1)$ is a
bounded join semilattice. Hence it remains to verify the absorption
and distributivity laws.
\begin{itemize}
  \item For $d(x)\cdot(d(x)+d(y))= d(x)$, we calculate
    \begin{align*}
      d(x)\cdot (d(y)+d(z))&= (d(y)+d(x))\cdot d(x)\\
& = d(y)\cdot d(x)+ d(x)\cdot d(x)\\
& = d(y)\cdot d(x) + d(x)\\
&=d(x)
    \end{align*}
    by commutativity and idempotence of meet as well as
    distributivity.
  \item $d(x)+d(x)\cdot d(y)= d(x)$. This is the last step of the
    previous proof.

  \item $d(x)\cdot (d(y)+d(z))=d(x)\cdot d(y)+d(x)\cdot d(z)$ is
    obvious from commutativity of meet and right distributivity.

  \item The distributivity law $d(x)+d(y)\cdot d(z)=(d(x)+d(y))\cdot
    (d(x)+d(z))$ holds by lattice duality.
\end{itemize}
\end{proof}
The next lemma presents additional domain laws; it is proved in
Appendix~2.
\begin{lemma}\label{P:domprops2}
In every dp-dioid,
\begin{enumerate}
 \item $x\le d(y)\cdot x \Leftrightarrow d(x)\le d(y)$,
\item $d(x)\cdot 0 = 0$,
\item $d(x)=0\Leftrightarrow x=0$,
\item $d(x)\le d(x+y)$.
\end{enumerate}
\end{lemma}
The \emph{least left preservation} law (1) is a characteristic
property of domain operations. It states that $d(x)$ is the least
domain element that satisfies the inequality $x\le p\cdot x$. Law (2)
shows that $0$ is a right annihilator in the subalgebra of domain
elements.

Next we consider the interaction between domain and the parallel operations.

\begin{lemma}\label{P:dmrdomprops}
In every dp-trioid,
\begin{enumerate}
  \item $d(1_\pi)=1_\sigma$,
\item $d(x\|y)=d(x)\|d(y)$,
\item $d(d(x)\|d(y)) = d(x)\|d(y)$,
\item $d(x)\|d(x)=d(x)$.
  \end{enumerate}
\end{lemma}
See Appendix~2 for proofs. By (3), the subalgebra of domain elements
is also closed with respect to parallel products, which are mapped to
meets. Property (4) follows from the fact that parallel products of
domain elements, hence of subidentities, are meets.

At the end of this section we characterise domain elements in terms of
a weak notion of complementation,
following~\cite{DesharnaisStruth11}. This further describes the
structure of domain elements within the subalgebra of subidentities.

\begin{proposition}\label{P:compprop}
  Let $S$ be a dp-dioid. Then $x\in d(S)$ if $x+y=1_\sigma$ and
  $y\cdot x=0$ hold for some $y\in S$.
\end{proposition}
\begin{proof}
  Fix $x$ and let $p$ be an element which satisfies $x+p=1_\sigma$ and
  $p\cdot x=0$. We must show that $d(x)=x$.
  \begin{itemize}
  \item  $x\cdot d(x)\le x$ since $d(x)\le 1_\sigma$ and $x = (x+p)\cdot x
  = x\cdot x = x\cdot d(x)\cdot x \le x\cdot d(x)$, whence $x\cdot
  d(x)=x$.

  \item $p\cdot d(x) = d(p\cdot d(x))\cdot p \cdot d(x) = d(p\cdot
  x)\cdot p\cdot d(x)= d(0)\cdot x \cdot d(x)= 0$.
  \end{itemize}
Therefore $d(x)=(x+p)\cdot d(x) = x\cdot d(x)+p\cdot d(x) = x$.
\end{proof}

We call an element $y$ of a dp-dioid a \emph{complement} of an element
$x$ whenever $x+y=1_\sigma$, $y\cdot x=0$ and $x\cdot y=0$ hold. Thus,
if $y$ is a complement of $x$, then $x$ is a completment of $y$. We
call an element \emph{complemented} if it has a complement. The set of
all complemented elements of a dp-dioid $S$ is denoted
$B_S$.

\begin{corollary}\label{P:comppropcor}
Let $S$ be a dp-dioid. Then $B_S\subseteq d(S)$.
\end{corollary}

\begin{lemma}\label{P:balemma}
  Let $S$ be a dp-dioid. Then $B_S$ is a boolean algebra.
\end{lemma}
\begin{proof}
  Since complemented elements are domain elements, they are idempotent
  and commutative. We use these properties to show that sums and
  products of complemented elements are complemented. More precisely,
  if $y_1$ is a complement of $x_1$ and $y_2$ a complement of
  $x_2$, then $y_1\cdot y_2$ is a complement of $x_1+x_2$ and
  $y_1+y_2$ a complement of $x_1\cdot x_2$. First,
\begin{align*}
  x_1+x_2+y_1\cdot y_2 &= x_1\cdot (x_2+y_2) + x_2\cdot (x_1+y_1) + y_1\cdot y_2\\
&= x_1\cdot x_2+x_1\cdot y_2+x_2\cdot x_1+x_2\cdot y_1+ y_1\cdot y_2\\
&= x_1\cdot x_2+x_1\cdot y_2+x_2\cdot y_1+ y_1\cdot y_2\\
&=(x_1+y_1)\cdot (x_2+y_2)\\
&=1_\sigma.
\end{align*}
Second, $(x_1+x_2)\cdot y_1\cdot y_2= x_1\cdot y_1\cdot y_2+x_2\cdot
y_1\cdot y_2= 0$. This proves complementation of sums. The proof of
complementation of products is dual, starting from $y_1\cdot y_2$.

These two facts show that $B_S$ is a subalgebra of $d(S)$. It is
therefore a bounded distributive sublattice and a boolean algebra,
since all elements are complemented and complements in distributive
lattices are unique.
\end{proof}

The following theorem summarises this investigation of the structure
of $d(S)$.

\begin{theorem}\label{P:greatestba}
  Let $S$ be a dp-dioid. Then $d(S)$ contains the greatest
  boolean subalgebra of $S$ bounded by $0$ and $1_\sigma$.
\end{theorem}

It is immediately clear that this theorem holds in dp-trioids as well.
In the abstract setting, it need not be the case that $d(S)$ contains
any Boolean algebra apart from $\{0,1_\sigma\}$. In fact, the
sequential subidentities may form a distributive lattice which is not
a boolean algebra, for instance a chain.

\begin{example}
  Consider the structure with addition defined by $0<a<1_\pi<1_\sigma$
  and the other operations defined by the following tables.
  \begin{equation*}
    \begin{array}{c|cccc}
      \cdot & 0 & a & 1_\pi & 1_\sigma\\
      \hline
0 & 0 & 0 & 0 &0\\
0 & a & a & a & a\\
1_\pi & 0 & a & 1_\pi & 1_\pi\\
1_\sigma & 0 & a & 1_\pi & 1_\sigma
    \end{array}
\qquad\qquad
\begin{array}{c|cccc}
  \| & 0 & a & 1_\pi & 1_\sigma\\
\hline
0 & 0 & 0 & 0 & 0\\
a & 0 & a & a & a\\
1_\pi & 0 & a & 1_\pi & 1_\sigma\\
1_\sigma & 0 & a & 1_\sigma & 1_\sigma
\end{array}
\qquad\qquad
    \begin{array}{c|c}
      &d\\
      \hline
      0&0\\
      a&a\\
      1_\pi&1_\sigma\\
      1_\pi&1_\pi
    \end{array}
  \end{equation*}
  It can be checked that this structure forms a dp-trioid (in fact
  this counterexample was found by Isabelle), but the elements $a$ and
  $1_\pi$ are not complemented. For instance, the only element $y$
  which satisfies $a+y=1_\sigma$ is $y=1_\sigma$, but $1_\sigma\cdot
  a= a\neq 0$.\qed
\end{example}
Thus $B_s$ need not be equal to $d(S)$, which justifies
Corollary~\ref{P:comppropcor}.  In the multirelational model, however,
the set of all sequential subidentities forms a boolean algebra, as
mentioned in Section~\ref{S:subidlaws}. In a multirelational dp-trioid
$S$, therefore, $d(S)=\{P \mid P\subseteq 1_\sigma\}$.

%%%%%%%%%%%%%%%%%%%%%%%%%%%%%%%%%%%%%%%%%%%%%%%%%%%%%%%%%%%%%%%%%%%%%%%%%%%%

\section{The Diamond Axioms of Star-Free CDL}\label{S:diaaxioms}

We are now equipped for deriving algebraic variants of  the diamond
axioms of concurrent dynamic logic except the star axioms in
dp-trioids. First, note that $\langle x\rangle p = \langle x\rangle
d(p)$.

\begin{lemma}\label{P:cdlaxioms1}~
\begin{enumerate}
\item In every dp-dioid, the following $\CDL$-axioms are derivable.
  \begin{enumerate}
\item $\langle x+y\rangle p = \langle x\rangle p+\langle y\rangle p$.
\item $\langle x\cdot y\rangle p = \langle x\rangle\langle y\rangle p$.
\item $\langle d(p)\rangle q = d(p)\cdot d(q)$.
\end{enumerate}
\item In every dp-trioid, the following $\CDL$-axiom is derivable as well.
  \begin{enumerate}
  \item[(d)] $\langle x\|y\rangle p = \langle x\rangle p\cdot \langle y\rangle p$.
  \end{enumerate}
\end{enumerate}
\end{lemma}
\begin{proof}
  \begin{enumerate}
  \item[(a)] Using right distributivity and additivity of domain, we
    calculate
    \begin{align*}
      \langle x+y\rangle p  &= d((x+y)\cdot p)\\
& = d(x\cdot p+ y\cdot p)\\
&=d(x\cdot p)+ d(y\cdot p)\\
&=\langle x\rangle p + \langle y\rangle p.
    \end{align*}
  \item[(b)] Using domain associativity and locality, we calculate
    \begin{align*}
      \langle x\cdot y\rangle p& = d((x\cdot y) \cdot p)\\
&= d(x\cdot (y \cdot d(p))\\
&=  d(x\cdot d(y \cdot d(p)))\\
&=  d(x\cdot \langle y \rangle p)\\
&=\langle x\rangle \langle y\rangle p.
    \end{align*}
  \item[(c)] By domain export,
$\langle d(p)\rangle q = d(d(p)\cdot q) = d(p)\cdot d(q)$.

  \item[(d)] Using domain interaction and the first domain concurrency
    axiom, we calculate
  \begin{align*}
    \langle x\| y\rangle p&
    = d((x\| y)\cdot d(p))\\
&= d((x\cdot d(p))\| (y\cdot d(p))\\
&= d(x\cdot d(p))\cdot d(y\cdot d(p))\\
&=\langle x\rangle p \cdot\langle y\rangle p.
  \end{align*}
  \end{enumerate}
\end{proof}
We can derive additional diamond laws from the domain laws such as
$\langle 0\rangle p = 0$ or $\langle 1_\sigma\rangle p = d(p)$. However, we
have a counterexample to $\langle 1_P\rangle p = 1_\sigma$, which holds in
the multirelational model.

\begin{example}
  Consider the structure with addition defined by $0< 1_\sigma< 1_p$,
  concurrent composition defined by meet, and the remaining operations
  by the conditions $1_\pi\cdot 0=0$, $1_\pi\cdot 1_\pi = 1_\pi$,
  $d(0)=0$ and $d(1_\sigma)=d(1_\pi)=1_\sigma$. It can be checked that
  this defines a dp-trioid, but $\langle 1_\pi\rangle 0 = d(1_\pi\cdot
  0)=d(0)=0< 1_\sigma$.\qed
\end{example}

The following \emph{demodalisation law} is proved in Appendix~2. It is
instrumental for deriving the star axioms of $\CDL$.
\begin{lemma}\label{P:demodalisation}
  In every dp-dioid, 
  \begin{equation*}
    \langle x\rangle p\le
  d(q)\Leftrightarrow x\cdot d(p)\le d(q)\cdot x.
  \end{equation*}
\end{lemma}

Finally we present two important counterexamples.
\begin{lemma}\label{P:diacounter}
  There are multirelations $R$, $P$ and $Q$ such that the following holds.
  \begin{enumerate}
  \item $\langle R\rangle (P \cup Q) \neq \langle R\rangle P\cup \langle R\rangle Q$,
  \item $\langle R\rangle \emptyset \neq \emptyset$. 
  \end{enumerate}
\end{lemma}
\begin{proof}
  \begin{enumerate}
  \item Let $R=\{(a,\{a,b\})\}$, $P=\{(a,\{a\})\}$ and $P=\{(b,\{b\})\}$. Then
    \begin{equation*}
      \langle R\rangle (P\cup Q) 
      = \{(a,\{a,b\})\}
      \supset \emptyset
=\langle R\rangle P\cup \langle R\rangle Q.
    \end{equation*}
  \item For $R=\{(a,\emptyset)\}$ we have $\langle R\rangle \emptyset =
    \{(a,\{a\})\}\neq \emptyset$.
  \end{enumerate}
\end{proof}
The additivity and strictness laws just refuted are defining
properties of modal algebras in the sense of J\'onsson and Tarski
(cf.~\cite{BlackburnRV01}). Our concurrent dynamic algebra axioms are
therefore nonstandard. This situation is analogous to the difference
between strict and multiplicative predicate transformers which arise
from relational semantics and their isotone counterparts which arise
from up-closed multirelations. Predicate transformers are usually
obtained from boxes instead of diamonds; the failure of
multiplicativity is related to that of additivity by duality.

In the concurrent setting, the above multirelation $R$ models an
external choice between $a$ and $b$ from input $a$. Reflecting this,
it is not sufficient that one can observe either one of $a$ and $b$,
but not both after executing $R$. In contrast to this, $\langle
S\rangle(P\cup Q)=\langle S\rangle P\cup\langle S\rangle Q$, for
$S=\{(a,\{a\}),(a,\{b\})\}$, models an internal choice.

%%%%%%%%%%%%%%%%%%%%%%%%%%%%%%%%%%%%%%%%%%%%%%%%%%%%%%%%%%%%%%%%%%%%%%%%%%%%

\section{The Star Axioms of CDL}\label{S:star}

This section derives the star axioms of CDL in expansions of dp-dioids
to variants of Kleene algebras. This is not entirely straightforward
due to the lack of associativity and left distributivity laws.  As
before we start at the level of multirelations to derive the
appropriate star axioms. We then lift the investigation to the
algebraic level.

Let $R$ and $S$ be multirelations.  Consider
the functions
\begin{equation*}
  F_{RS} = \lambda X.\ S\cup R\cdot X,\qquad F_{R} = \lambda X.\ 1_\sigma\cup R\cdot X,
\end{equation*}
which generate variants of the Kleene star as their least
fixpoints. Existence of these fixpoints is guaranteed by basic
fixpoint theory. The universal multirelation $U$ has been introduced
in Section~\ref{S:multirelations}.
\begin{lemma}\label{P:fixpointprops}~
  \begin{enumerate}
  \item The functions $F_{RS}$ and $F_{R}$ are isotone.
  \item $(M(X),\cup,\cap,\emptyset,U)$ forms a complete lattice.
  \item $F_{RS}$ and $F_{R}$ have least pre-fixpoints and greatest
    post-fixpoints which are also least and greatest fixpoints.
  \end{enumerate}
\end{lemma}
See Appendix~2 for proofs.

We write $(R^\ast S)$ or $\mu F_{RS}$ for the least fixpoint of
$F_{RS}$ and $R^\ast$ or $\mu F_R$ for the least fixpoint of $F_R$. We
immediately obtain the fixpoint unfold and induction laws
  \begin{equation*}
    S\cup R\cdot (R^\ast S) \subseteq (R^\ast S),\qquad S\cup R\cdot T\subseteq T \Rightarrow (R^\ast S) \subseteq T
  \end{equation*}
  for $F_{RS}$ and the corresponding laws
  \begin{equation*}
    1_\sigma\cup R\cdot R^\ast \subseteq R^\ast S,\qquad 1_\sigma\cup R\cdot T\subseteq T \Rightarrow R^\ast \subseteq T
  \end{equation*}
  for $F_{R}$. The binary fixpoint $(R^\ast S)$ is not necessarily
  equal to $R^\ast\cdot S$. At least, by definition, $R^\ast = R^\ast
  \cdot 1_\sigma = (R^\ast 1_\sigma)$. The fixpoints $\mu F_R$ and
  $\mu F_{RS}$ can be related by the following well known \emph{fixpoint
    fusion} law.
\begin{theorem}\label{P:fixpointfusion}~
  \begin{enumerate}
  \item Let $f$ and $g$ be isotone functions and $h$ a continuous
    function over a complete lattice. If %$h(0)\le \mu g$ and
    $h\circ g\le f\circ h$, then $h(\mu g)\le \mu f$.
  \item Let $f$, $g$ and $h$ be isotone functions over a complete
    lattice. If $f\circ h \le h \circ g$, then $\mu f\le h(\mu g)$.
  \end{enumerate}
\end{theorem}
It follows from (1) and (2) that, if $f$ and $g$ are isotone, $h$ is
continuous and $h\circ g = f \circ h$, then $\mu f = h(\mu
g)$. Applying fixpoint fusion to $F_{RS}$ and $F_R$ yields the
following fact.
\begin{corollary}\label{P:fixpointfusioncor}
  Let $R$, $S$ and $T$ be multirelations. Then 
\begin{equation*}
R^\ast \cdot
  S\subseteq (R^\ast S),\qquad (R^\ast S) \cdot T \subseteq (R^\ast
  (S\cdot T)).
\end{equation*}
\end{corollary}
\begin{proof}
  Let $f=F_{RS}$, $g=F_{R}$ and $h=H =\lambda X. X\cdot S$. 

  It is easy to show that $H$ is continuous, that is, $(\bigcup_{i\in
    I} R_i)\cdot S = \bigcup_{i\in I} (R_i\cdot S)$. The proof is
  similar to that of Lemma~\ref{P:seqlaws}(4). Moreover
\begin{align*}
  (H\circ F_R)(x)  
&= (1_\sigma\cup R\cdot x)\cdot S\\
&= S\cup (R\cdot x)\cdot S
 \subseteq S\cup R\cdot (x\cdot S)\\
&= (F_{RS}\circ H)(x)
\end{align*}
by weak associativity (Lemma~\ref{P:seqlaws}(3)), so $R^\ast \cdot S
\subseteq (R^\ast S)$ by fixpoint fusion.

The proof of $(R^\ast S)\cdot T\subseteq (R^\ast (S\cdot T))$ follows
the same pattern.
\end{proof}
Proving the converse direction, $(R^\ast S) \subseteq R^\ast\cdot S$,
by fixpoint fusion requires associativity in the other direction,
which does not hold in our setting (Lemma~\ref{P:counterexamples}(2),
where the counterexample was given for $R\cdot (R\cdot S)\subseteq
(R\cdot R)\cdot S$ and extends to the case above). The following
counterexample rules out any other proof of this inclusion.

\begin{lemma}\label{P:fixcounter}
  There are multirelations $R$ and $S$ such that $R^\ast S\neq
  R^\ast\cdot S$.
\end{lemma}
\begin{proof}
  Consider $R$ and $S$ from Lemma~\ref{P:counterexamples}(2) and their diagrams in Figure~\ref{Fig:three}.
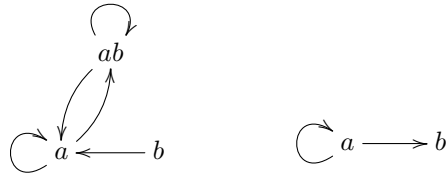
\begin{figure}[h]
  \centering

\begin{equation*}
  \def\labelstyle{\normalsize}
\xymatrix@C=5pt{
&ab\ar@/_/[dl]\ar@(ul,ur)&\\
a\ar@/_/[ur]\ar@(dl,ul)&&b\ar[ll]\\
}
\qquad\qquad\qquad
\xymatrix{
&\\
a\ar[r]\ar@(dl,ul)&b
}
\end{equation*}
\caption{Diagrams for $R$ and $S$ in the proof of Lemma~\ref{P:counterexamples}(2) (same as Fig.~\ref{Fig:one})}
\label{Fig:three}
\end{figure}
The multirelations $R^\ast = 1_\sigma\cup R\cdot (1_\sigma\cup R)$, $R^\ast\cdot
S$ and $R^\ast S= S\cup R\cdot (S\cup R\cdot (S\cup R)) = R\cdot
R\cdot (R\cup S)$ are computed from these diagrams as shown in Figure~\ref{Fig:four}.
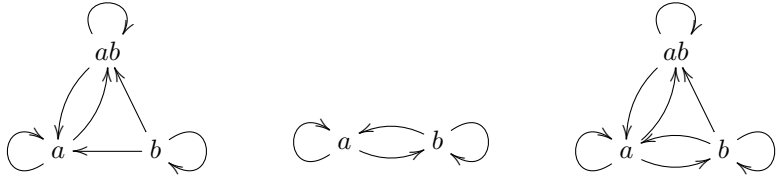
\begin{figure}[h]
  \centering
\begin{equation*}
   \def\labelstyle{\normalsize}
\xymatrix@C=5pt{
&ab\ar@/_/[dl]\ar@(ul,ur)&\\
a\ar@/_/[ur]\ar@(dl,ul)&&b\ar[ll]\ar[ul]\ar@(ur,dr)
}
\qquad\qquad\qquad
\xymatrix{
&\\
a\ar@/_/[r]\ar@(dl,ul)&b\ar@(ur,dr)\ar@/_/[l]
}
\qquad\qquad\qquad
\xymatrix@C=5pt{
&ab\ar@/_/[dl]\ar@(ul,ur)&\\
a\ar@/_/[ur]\ar@(dl,ul)\ar@/_/[rr]&&b\ar@/_/[ll]\ar[ul]\ar@(ur,dr)
}
\end{equation*}
 
  \caption{Diagrams for $R^\ast$, $R^\ast\cdot S$ and $R^\ast S$ for $R$ and $S$ in the proof of Lemma~\ref{P:counterexamples}(2)}
  \label{Fig:four}
\end{figure}
Clearly, $R^\ast S\not\subseteq R^\ast\cdot S$.
\end{proof}
At first sight, Lemma~\ref{P:fixcounter} seems to invalidate the
star-axiom of $\CDL$. However, the identity $R^\ast S = R^\ast\cdot S$
is only needed in the modal setting, where $S$ is a subidentity. In
this case, as we have seen, stronger algebraic properties for
sequential composition are present. We now investigate this
restriction.

First, we show that the unfold law for $R^\ast S$ can be strengthened
to an identity.
\begin{corollary}\label{P:mrunfold}
  Let $R$ and $S$ be multirelations. Then $S\cup R\cdot (R^\ast S) =
  (R^\ast S)$.
\end{corollary}
  This holds since every pre-fixpoint of $F_{RS}$ is also a fixpoint.

  We now prove the desired fusion of $\mu F_{RS}$ with $\mu F_R$ when
  $S$ is a subidentity.
\begin{proposition}\label{P:subidfusion}
  Let $R$ be a multirelation and $P$ a subidentity. Then
  \begin{equation*}
R^\ast P = R^\ast\cdot P. 
\end{equation*}
\end{proposition}
\begin{proof}
  Applying fixpoint fusion as in Corollary~\ref{P:fixpointfusioncor},
  but with $H=\lambda X. X\cdot P$, now establishes $H\circ F_R=
  F_{RS}\circ H$, since we have full associativity for subidentities
  by Lemma~\ref{P:subidlaws}(1). This suffices to verify the claim.
\end{proof}
We can therefore replace $R^\ast P$ by $R^\ast \cdot P$ in the
induction law for $F_{RS}$.
\begin{lemma}\label{P:mrinduction}
Let $R$ and $S$ be multirelations and $P$ be a subidentity. Then
\begin{equation*}
P\cup R\cdot S \subseteq S\Rightarrow R^\ast \cdot P\subseteq S.
\end{equation*}
\end{lemma}
Corollary~\ref{P:mrunfold} and Lemma~\ref{P:mrinduction} motivate the
following algebraic definition. As before we use domain elements instead of sequential subidentities.

A \emph{proto-Kleene algebra with domain} (\emph{dp-Kleene algebra})
is a dp-dioid expanded by a star operation which satisfies the
\emph{(left) star unfold} and \emph{(left) star induction} axioms
\begin{equation*}
  1_\sigma+x\cdot x^\ast \le x^\ast,\qquad d(z)+x\cdot y\le y\Rightarrow x^\ast\cdot d(z)\le y.
\end{equation*}
A \emph{proto-bi-Kleene algebra with domain} (\emph{dp-bi-Kleene
  algebra}) is a dp-Kleene algebra which is also a
dp-trioid\footnote{In this article we ignore the star of concurrent
  composition, which should normally be part of the definition of a
  bi-Kleene algebra. The reason is that it is not considered in
  $\CDL$.}. In both cases, the unfold law can be strengthened to the
identity $ 1_\sigma+x\cdot x^\ast = x^\ast$. The full list of
dp-bi-Kleene algebra axioms can be found in Appendix~1.

The development so far is summarised in the following soundness
result, which links the multirelational layer with the abstract
algebraic one.

\begin{theorem}\label{P:dmrkamodel}
  $(M(X),\cup,\cdot, \|, \emptyset,1_\sigma,1_\pi,d,^\ast)$ is a
  dp-bi-Kleene algebra.
\end{theorem}
\begin{proof}
  The structure is a dp-trioid as a consequence of
  Theorem~\ref{P:dmramrmodel}. The star axioms hold by
  Corollary~\ref{P:mrunfold} and Lemma~\ref{P:mrinduction}.
\end{proof}

Due to this result we can now continue at the algebraic level. First
we derive the modal star unfold axiom of $\CDL$.
\begin{lemma}\label{P:cdaunfold}
  Let $K$ be a dp-Kleene algebra, $x\in K$ and $p\in
  d(K)$. Then
\begin{equation*}
  p+\langle x\rangle\langle x^\ast\rangle p = \langle x^\ast\rangle p.
\end{equation*}
\end{lemma}
\begin{proof}
  $p+\langle x\rangle\langle x^\ast\rangle p = \langle 1_\sigma+x\cdot
  x^\ast\rangle p =\langle x^\ast\rangle p$ by the star unfold axiom
  and the $\CDL$ axioms which have been verified in
  Lemma~\ref{P:cdlaxioms1}.
\end{proof}
It remains to verify the star induction axiom of $\CDL$. First we show
a simulation law.
\begin{lemma}\label{P:simulation}
  Let $K$ be a dp-Kleene algebra,  $x\in K$ and
  $p\in d(K)$. Then
\begin{equation*}
x\cdot p \le p\cdot y\Rightarrow x^\ast \cdot
  p\le p\cdot y^\ast.
\end{equation*}
\end{lemma}
See Appendix~2 for a proof.  The derivation of an
algebraic variant of the star unfold axiom of $\CDL$ is then trivial.
\begin{proposition}\label{P:cdainduct}
  Let $K$ be a dp-Kleene algebra, $x\in K$ and $p\in d(K)$. Then
\begin{equation*}
\langle x\rangle p \le p \Rightarrow \langle x^\ast\rangle p \le p.
\end{equation*}
\end{proposition}
\begin{proof}
  \begin{equation*}
    \langle x\rangle p \le q \Leftrightarrow x\cdot p \le p\cdot x \Rightarrow x^\ast \cdot p \le p \cdot x^\ast \Leftrightarrow \langle x^\ast\rangle p \le p.
  \end{equation*}
  The first and last step use demodalisation
  (Lemma~\ref{P:demodalisation}), the second step uses
  Lemma~\ref{P:simulation}.
\end{proof}

The first main theorem of this article combines these results.

\begin{theorem}\label{P:cda}
  The $\CDL$ axioms are derivable in dp-bi-Kleene algebras.
\end{theorem}

We therefore call dp-bi-Kleene algebras informally \emph{concurrent
  dynamic algebras}.

Finally, in Appendix~2, we prove a right star unfold law and derive a
variant of modal star induction in analogy to the induction axiom of
pd-Kleene algebra.
\begin{lemma}\label{P:starvar}
  Let $K$ be a pd-Kleene algebra, $x\in K$ and $p,q\in d(K)$. Then
\begin{enumerate}
\item $p+\langle x^\ast\rangle\langle x\rangle p \le \langle
  x^\ast\rangle p$,
\item $p+ \langle x\rangle q \le q \Rightarrow \langle x^\ast
  \rangle p \le q$.
\end{enumerate}
\end{lemma}

%%%%%%%%%%%%%%%%%%%%%%%%%%%%%%%%%%%%%%%%%%%%%%%%%%%%%%%%%%%%%%%%%%%%%%%%%%%%%

\section{The Structure of AP-Trioids}\label{S:amrdioids}

Section~\ref{S:dmrdioids} shows that the domain elements of a dp-dioid
or dp-trioid form a distributive lattice. We now revisit this
development for antidomain, where the resulting domain algebras are
boolean algebras. We start with a number of auxiliary lemmas. These
are needed because the minimality of the axiom set makes it difficult
to derive the desirable properties directly.

In the following lemma we abbreviate $d=a\circ a$. This is justified
in Proposition~\ref{P:antidomaindomain}, which formally verifies that
$a(a(x))$ models the domain of element $x$.

\begin{lemma}\label{P:antidomprops}
In every ap-dioid,
  \begin{enumerate}
  \item $a(x)\le 1_\sigma$,
\item $a(x)\cdot a(x)=a(x)$,
\item $a(x)=1_\sigma \Leftrightarrow x =0$,
\item $a(x)\cdot y = 0 \Leftrightarrow  a(x)\le a(y)$,
\item $x \le y \Rightarrow a(y)\le a(x)$,
\item $a(x)\cdot a(y)\cdot d(x+y)=0$,
\item $a(x+y) = a(x)\cdot a(y)$,
\item $a(a(x)\cdot y) = d(x)+a(y)$.
\end{enumerate}
\end{lemma}
See Appendix~2 for proofs. The \emph{greatest left
  annihilation property} (4) is a characteristic property of antidomain
elements. It states that $a(x)$ is the greatest antidomain element $p$
which satisfy the left annihilation law $p\cdot x = 0$. By (5), the
antidomain operation is \emph{antitone}; by (7) it is
\emph{multiplicative}. Property (8) is an \emph{export} law for
antidomain. These laws are helpful in the following proposition which is
proved in Appendix~2.

\begin{proposition}\label{P:antidomaindomain}
  Every ap-dioid is a dp-dioid with domain operation $d=a\circ a$.
\end{proposition}

As in Section~\ref{S:dmrdioids}, we investigate the structure of
domain elements.

\begin{proposition}\label{P:antidomainba}
  Let $S$ be an ap-dioid with $d=a\circ a$. Then $d(S)$
  forms a subalgebra which is the greatest boolean algebra in $S$
  bounded by $0$ and $1$.
\end{proposition}
\begin{proof}
  First, since every ap-dioid is a dp-dioid, $d(S)$ is a bounded
  distributive lattice. Second, antidomain elements are closed under
  the operations because $d(a(x))=a(x)$: by antidomain locality,
  \begin{equation*}
    d(a(x))=a(a(a(x)))=a(d(x))= a(1_\sigma\cdot d(x))= a(1_\sigma\cdot x)=a(x).
  \end{equation*}
  Third, the operation $\lambda x.a(x)$ is complementation in this
  algebra.  One of the complementation properties,
  $a(d(x))+d(x)=a(x)+d(x)=1_\sigma$, is an axiom. The other ones,
  $a(d(x))\cdot d(x)=a(x)\cdot d(x)=0$ and $d(x)\cdot
  a(d(y))=d(x)\cdot a(x)=0$, are immediate from antidomain
  annihilation.

  Finally, by Theorem~\ref{P:greatestba}, $d(S)$ contains the greatest
  boolean algebra in $S$ between $0$ and $1_\sigma$ and is therefore equal to
  the greatest such boolean algebra.
\end{proof}

We now expand Proposition~\ref{P:antidomaindomain} from the sequential
to the concurrent case.

\begin{proposition}\label{P:amrdmr}
  Every ap-trioid is a dp-trioid.
\end{proposition}
The proof can be found in Appendix~2.

Finally we investigate the star. A \emph{proto-Kleene algebra with
  antidomain} (\emph{ap-Kleene algebra}) is an ap-dioid expanded by a
star operation which satisfies the \emph{(left) star unfold} and
\emph{(left) star induction} axioms
\begin{equation*}
  1_\sigma+x\cdot x^\ast \le x^\ast,\qquad a(z)+x\cdot y\le y\Rightarrow x^\ast\cdot a(z) \le y.
\end{equation*}
A \emph{proto-bi-Kleene algebra with antidomain} (\emph{ap-bi-Kleene
  algebra}) is an ap-Kleene algebra which is also an ap-trioid. A full
list of ap-bi-Kleene algebra axioms can be found in Appendix~1.

The following propsition is immediate from
Propositions~\ref{P:antidomaindomain} and~\ref{P:amrdmr}.
\begin{proposition}\label{P:KAdomainantidomain}
  \begin{enumerate}
  \item Every ap-Kleene algebra is a dp-Kleene.
  \item Every ap-bi-Kleene algebra is a dp-bi-Kleene algebra.
  \end{enumerate}
\end{proposition}

In combination, these facts establish an analogon to
Theorem~\ref{P:dmrkamodel}.
\begin{theorem}\label{P:amrkamodel}
  $(M(X),\cup,\cdot,\|,\emptyset,1_\sigma,1_\pi,a,^\ast)$ forms an
  ap-bi-Kleene algebra.
\end{theorem}

%%%%%%%%%%%%%%%%%%%%%%%%%%%%%%%%%%%%%%%%%%%%%%%%%%%%%%%%%%%%%%%%%%%%%%%%%%

\section{The Box Axioms of CDL}\label{S:boxaxioms}

The results of the previous section imply that the diamond axioms of
concurrent dynamic logic hold in the setting of antidomain
algebras. In addition we can now derive algebraic variants of Peleg's
De Morgan dual box axioms. Since every ap-bi-Kleene algebra is a
dp-Kleene algebra, the diamond axioms of concurrent dynamic algebras
hold immediately.

\begin{lemma}\label{P:antidomaindiacda}~
\begin{enumerate}
\item In every ap-dioid, the following CDL-axioms are derivable.
  \begin{enumerate}
\item $\langle x+y\rangle p = \langle x\rangle p+\langle y\rangle p$.
\item $\langle x\cdot y\rangle p = \langle x\rangle\langle y\rangle p$.
\item $\langle d(p)\rangle q = d(p)\cdot d(q)$.
\end{enumerate}
\item In every ap-trioid, the following CDL-axiom is derivable.
  \begin{enumerate}
  \item[(d)] $\langle x\|y\rangle p = \langle x\rangle p\cdot \langle y\rangle p$.
  \end{enumerate}
\item In every ap-Kleene algebra, the following star axioms are derivable.
\begin{enumerate}
\item[(e)] $1_\sigma+\langle x\rangle\langle x^\ast\rangle p = \langle x^\ast\rangle p$.
\item[(f)] $\langle x\rangle p \le p \Rightarrow \langle x^\ast\rangle p \le p$.
\end{enumerate}
\end{enumerate}
\end{lemma}
In addition, the following box axioms follow easily from De Morgan
duality.
\begin{proposition}\label{P:antidomainboxcda}~
\begin{enumerate}
\item In every ap-dioid, the following CDL-axioms are derivable.
  \begin{enumerate}
\item $[x+y] p = [x]p\cdot [y]p$.
\item $[x\cdot y]p = [x][y] p$.
\item $[d(p)] q = a(p)+ d(q)$.
\end{enumerate}
\item In every ap-trioid, the following CDL-axiom is derivable.
  \begin{enumerate}
  \item[(d)] $[x\|y] p = [x] p\cdot [y] p$.
  \end{enumerate}
\item In every ap-Kleene algebra, the following star axioms are derivable.
\begin{enumerate}
\item[(e)] $1_\sigma\cdot [x][x^\ast]p = [x^\ast]p$.
\item[(f)] $ p \le [x]p \Rightarrow p \le [x^\ast] p$.
\end{enumerate}
\end{enumerate}
\end{proposition}
In sum, these results yield the second main theorem of this article.
\begin{theorem}\label{P:cdaboxdiamond}
  The box and diamond axioms of $\CDL$ are derivable in ap-bi-Kleene
  algebras.
\end{theorem}
We therefore call ap-bi-Kleene algebras \emph{concurrent dynamic
  algebras} as well. In contrast to dp-Kleene algebras, these are
based on boolean algebras of domain elements.

Finally we present counterexamples to multiplicativity and
co-strictness of boxes.

\begin{lemma}\label{P:boxcounter}
  There are multirelations $R$, $P$ and $Q$ such that the following holds.
  \begin{enumerate}
  \item $[R](P\cdot Q) \neq [R]P \cdot [R]Q$,
  \item $[R]1_\sigma \neq 1_\sigma$.
  \end{enumerate}
\end{lemma}
\begin{proof}~
  \begin{enumerate}
  \item Obviously, $\forall p, q.\ \langle x\rangle (p+q) = \langle
    x\rangle p+\langle x\rangle q$ if and only if $\forall p,q.\
    [x](p\cdot q) = [x]p\cdot [x]q$. Hence the counterexample from
    Lemma~\ref{P:diacounter} applies. 
\item Similarly, $\langle x\rangle 0 = 0$ if and only if $[x]1_\sigma=1_\sigma$.
  \end{enumerate}
\end{proof}

The following counterexample is directly related to this
lemma. According to J\'onsson and Tarski, modal boxes and diamonds are
conjugate functions on boolean algebras, that is, they are related by
the conjugation law
\begin{equation*}
  \langle x\rangle p\cdot q = 0\Leftrightarrow p\cdot [x]q=0.
\end{equation*}
Conjugate functions are a fortiori additive. By
Lemma~\ref{P:diacounter} and~\ref{P:boxcounter}, this cannot be the
case in the multirelational setting, hence the conjugation law cannot
hold. This is confirmed directly by the multirelation $R
=\{(a,\emptyset)\}$ and the subidentity $P= \{(a,\{a\})\}$ over the
set $X=\{a\}$, which satisfy
\begin{align*}
  \langle R\rangle P \cdot P =d(R\cdot P)\cdot P= P \supset \emptyset
  = P\cdot a(R\cdot a(P)) = P \cdot [R]P).
\end{align*}

%%%%%%%%%%%%%%%%%%%%%%%%%%%%%%%%%%%%%%%%%%%%%%%%%%%%%%%%%%%%%%%%%%%%%%%%%%%%%%%%
\section{The Star and Finite Iteration}\label{S:finiteiteration}

It is well known that least fixpoints can be reached by iterating from
the least element of a complete lattice up to the first ordinal
whenever the function under consideration is not only isotone, but
also continuous. Otherwise, if the function is only isotone,
transfinite induction beyond the first ordinal is required.

Our counterexample to left distributivity rules out continuity in
general, but, in fact, chain continuity or directedness suffices for
the star.  As in Section~\ref{S:star}, we consider
\begin{equation*}
  F_R=\lambda X.1_\sigma\cup R\cdot X.
\end{equation*}

Peleg has provided a counterexample even to chain
completeness~\cite{Peleg87}. We display a proof in
Appendix~2 to make this article selfcontained.

\begin{lemma}[Peleg]\label{P:chaincompletecounter}
  There exists a multirelation $R$ and an ascending chain of
  multirelations $S_i$, $i\in\mathbb{N}$, such that $F_R(\bigcup_{i\in\mathbb{N}}S_i) \neq \bigcup_{i\in\mathbb{N}}F_R(S_i)$.
\end{lemma}
Chain completeness can, however, be obtained if a multirelation $R$ is
\emph{externally image finite}, that is, for all $(a,A) \in R$ the set
$A$ has finite cardinality. This notion has been called \emph{finitely
  branched} by Peleg. We have chosen a different name to distinguish
it from \emph{interal image finiteness}, which is the case when for
each $a$, the set of all $(a,A)$ has finite cardinality. From a
computational point of view, external image finiteness is not a
limitation, since infinite sets $A$ correspond to unbounded external
nondeterminism or unbounded concurrent composition, which is not
implementable.

\begin{lemma}[Peleg]\label{P:chaincont}
  If $R$ is externally image finite, then $F_R$ is chain continuous.
\end{lemma}
See Appendix~2 for a proof.

We define powers of $F_R$ inductively as $F_R^0=\lambda X.X$ and
$F_R^{n+1} = F_R\circ F_R^n$ and can then define iteration to the
first limit ordinal as
\begin{equation*}
  F_R^\ast = \bigcup_{i\in\mathbb{N}}F^i.
\end{equation*}
General fixpoint theory (Kleene's fixpoint theorem) then implies the
following fact.

\begin{proposition}\label{P:iterprop}
  If $R$ is externally image finite, then $R^\ast =
  F_R^\ast(\emptyset)$.
\end{proposition}

We now compare this notion of finite iteration with another one.
\begin{gather*}
  R^{(0)} = \emptyset,\qquad R^{(n+1)} = 1_\sigma\cup R\cdot R^{(n)},\qquad  R^{(\ast)} = \bigcup_{n\in\mathbb{N}} R^{(n)}.
\end{gather*}
Our next lemma shows that the inductive definition of $R^{(\ast)}$
captures the iterative function application of $F^\ast_R$ to
$\emptyset$ and hence $R^\ast$ for external image finiteness. It is
proved in Appendix~2.
\begin{lemma}\label{P:iterzeroone}~
\begin{enumerate}
\item  For all $n$, $F_R^n(\emptyset) = R^{(n)}$ and therefore $F_R^\ast(\emptyset)=R^{(\ast)}$.
\item If $R$ is externally image finite, then $R^\ast = R^{(\ast)}$.
\end{enumerate}
\end{lemma}
Finally we show that external image finiteness in
Lemma~\ref{P:iterzeroone}(2) is nessesary.

\begin{lemma}\label{P:itercounter}
  There exists a multirelation $R$ such that $R^{(\ast)}$ is not a
  fixpoint of $F_R$.
\end{lemma}
\begin{proof}
Consider the multirelation
\begin{equation*}
  R = \{ (m, \{n\ | \ n < m\}) \mid m\in\mathbb{N}\cup\{\infty\}\}.
\end{equation*}
It follows that $(0,\emptyset) \in R$ and $R\cdot\emptyset =
\{(0,\emptyset)\}$.

Then $(m,\{n\mid n\le m-2\})\not\in R$ but it is in $R^{(2)}$, and
$(m,\{n\mid n\le m-k\})\not\in R^{(i)}$ for $i<k$, but it is in
$R^{(k)}$; similarly $(m,\emptyset)\in R^{(m)}$ but not in $R^{(l)}$
for all $l<m$. Consequently, $(\infty,\emptyset)\not\in R^{(n)}$ for
all $n\in\mathbb{N}$, and therefore $(\infty,\emptyset)\not\in
R^{(\ast)}$, but $(\infty,\emptyset)\in F_R(R^{(\ast)})$.
\end{proof}

%%%%%%%%%%%%%%%%%%%%%%%%%%%%%%%%%%%%%%%%%%%%%%%%%%%%%%%%%%%%%%%%%%%%%%%%%%%%%%

\section{Refutation of Segerberg's Axiom}\label{S:segerberg}

Segerberg's axiom is the induction axiom of (non-concurrent)
propositional dynamic logic (cf.~\cite{HarelKozenTiuryn}). Goldblatt
uses its box version---his box semantics is different from ours---but
not the diamond one. This section provides a counterexample to
Segerberg's axiom in the multirelational model with box-diamond
duality.

In diamond form, Segerberg's axiom is
\begin{equation*}
  \langle x^\ast\rangle p \le p + \langle x^\ast\rangle (\langle x\rangle p - p).
\end{equation*}
In modal Kleene algebra it is equivalent to the star induction
axiom. For multirelations, the situation is different.

\begin{proposition}\label{P:segerbergcounter}
  There is a multirelation $R$ and a subidentity $P$ such that
  \begin{equation*}
    \langle R^\ast\rangle P \supset P \cup \langle R^\ast\rangle (\langle R\rangle P - P).
  \end{equation*}
\end{proposition}

\begin{proof}
  Let $R = \{(a,\{b,c\}),(b,\{b\}),(b,\{c\}),(c,\{c\})\}$ and $P =
  \{(c,\{c\})\}$. As previously, we visualise $R$ in the Hasse diagram
  in Figure~\ref{Fig:five}. The multirelation $R^\ast$ can be read off
  as the relational reflexive transitive closure from this diagram by
  chasing arrows.
  \begin{figure}[h]
    \centering
 \begin{equation*}
  \def\labelstyle{\normalsize}
\xymatrix{
bc\ar@(ul,ur)\ar[drr]&&\\
a\ar[u]&b\ar@(dl,dr)\ar[r]&c\ar@(ul,ur)\\
}
\end{equation*}
\vspace{.3cm}
     \caption{Diagram for $R$ in the proof of Proposition~\ref{P:segerbergcounter}}
    \label{Fig:five}
  \end{figure}
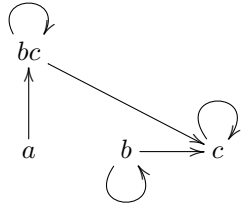
One can also use the diagram to check that
\begin{align*}
  R\cdot P &=\{(b,\{c\}),(c,\{c\})\},\\
  \langle R\rangle P &= \{(b,\{b\}),(c,\{c\})\},\\
\langle R\rangle P - P &= \{(b,\{b\})\}.
\end{align*}
One can compute $R^\ast$ by iterating with $R^{(\ast)}$
according to Lemma~\ref{P:iterzeroone}(2), since $R$ is externally
image finite. Obviously, $R\cdot \emptyset=\emptyset$. Therefore,
\begin{align*}
R^{(1)} &= 1_\sigma,\\
R^{(2)} &=  1_\sigma\cup R\cdot (1_\sigma\cup R)\\
& = \{(a,\{a\}),(a,\{c\}),(a,\{b,c\}),(b,\{b\}),(b,\{c\}),(c,\{c\})\},\\
R^{(3)} &= 1_\sigma\cup R\cdot (1_\sigma\cup R\cdot (1_\sigma\cup R)) = R^{(2)},\\
R^{(n)} &= R^{(2)},
\end{align*}
that is, iteration becomes stationary after four steps. Chain
completeness implies that
\begin{equation*}
 R^\ast = R^{(\ast)}= R^{(2)} = \{(a,\{a\}),(a,\{c\}),(a,\{b,c\}),(b,\{b\}),(b,\{c\}),(c,\{c\})\}.
\end{equation*}
On the one hand, his result yields
\begin{align*}
  \langle R^\ast\rangle (\langle R\rangle P - P) &= \langle R^\ast\rangle \{(b,\{b\})\}= \{(b,\{b\})\},\\
P\cup \langle R^\ast\rangle (\langle R\rangle P - P) &= \{(b,\{b\}),(c,\{c\})\}.
\end{align*}
On the other hand we obtain
\begin{align*}
R^\ast \cdot P &= \{(a,\{c\}),(b,\{c\}),(c,\{c\})\},\\
  \langle R^\ast\rangle P &= \{(a,\{a\}),(b,\{b\}),(c,\{c\})\}.
\end{align*}
This confirms that $\langle R^\ast\rangle P \supset P \cup \langle
R^\ast\rangle (\langle R\rangle P - P)$ and falsifies Segerberg's
formula.
\end{proof}
\begin{corollary}\label{P:segerbergcor}
Segerberg's axiom is not derivable in ap-bi-Kleene algebras.
\end{corollary}
Obviously, this implies that the axiom is not derivable in ap-Kleene
algebras. However, at least its converse is derivable.
\begin{lemma}\label{P:segerbergconv}
In every ap-Kleene algebra,
\begin{equation*}
p + \langle x^\ast\rangle (\langle x\rangle p - p) \le \langle x^\ast\rangle p.
\end{equation*}
\end{lemma}
See Appendix~2 for a proof. Hence this fact is derivable in
ap-bi-Kleene algebras, too.

Segerberg's axiom is usually presented in box form as $p\cdot
[x^\ast](p\to [x]p) \le [x^\ast]p$, where $p\to q = a(p)+q$. By De
Morgan duality, variants of Proposition~\ref{P:segerbergcounter},
Corollary~\ref{P:segerbergcor} and Lemma~\ref{P:segerbergconv} hold in
the box case. In particular, the box variant of Segerberg's axiom is
neither valid in the multirelational model nor derivable in ap-bi-Kleene
algebras.

%%%%%%%%%%%%%%%%%%%%%%%%%%%%%%%%%%%%%%%%%%%%%%%%%%%%%%%%%%%%%%%%%%%%%%%%%%%%%%

\section{Conclusion}\label{S:conclusion}

We have defined weak variants of Kleene algebras with domain and
antidomain which capture essential properties of the algebra of
multirelations under union, sequential and concurrent composition and
the sequential Kleene star together with multirelational domain and
antidomain operations. The relationships between the different
algebraic structures defined in this article is summarised in
Figure~\ref{F:structure}. Both dp-bi-Kleene alegebras and ap-bi-Kleene
algebras qualify as concurrent dynamic algebras; their axioms are
listed in Appendix~1.
\begin{figure}[h]
  \centering
 \begin{equation*}
   \def\labelstyle{\normalsize}
   \xymatrix
   @!=.5cm
   { 
     &&&&\text{$\mathsf{apbKA}$}\\
     &&& \text{$\mathsf{apT}$}\ar@{-}[ur]&&\text{$\mathsf{dpbKA}$}\ar@{-}[ul]\\
     &\text{$\mathsf{apKA}$}\ar@{-}[uurrr]  &&&\text{$\mathsf{dpT}$}\ar@{-}[ul]\ar@{-}[ur]\\
     \text{$\mathsf{apD}$}\ar@{-}[ur]\ar@{-}[uurrr] && \text{$\mathsf{dpKA}$}\ar@{-}[ul]\ar@{-}[uurrr]&&&\text{$\mathsf{pT}$}\ar@{-}[ul]\\
     &\text{$\mathsf{dpD}$}\ar@{-}[ul]\ar@{-}[ur]\ar@{-}[uurrr]\\
     &&\text{$\mathsf{pD}$}\ar@{-}[ul]\ar@{-}[uurrr]
   }
\end{equation*} 
\caption{Summary of algebraic subclass relationships. $\mathsf{pD}$
  stands for the class of proto-dioids, $\mathsf{dpD}$ for domain
  proto-dioids, $\mathsf{apD}$ for antidomain proto-dioids,
  $\mathsf{dpKA}$ for dp-Kleene algebras, $\mathsf{apKA}$ for
  ap-Kleene algebras, $\mathsf{pT}$ for proto-trioids, $\mathsf{dpT}$
  for dp-trioids, $\mathsf{apT}$ for ap-trioids, $\mathsf{dpbKA}$ for
  dp-bi-Kleene algebras and $\mathsf{apbKA}$ for ap-bi-Kleene
  algebras.}
  \label{F:structure}
\end{figure}
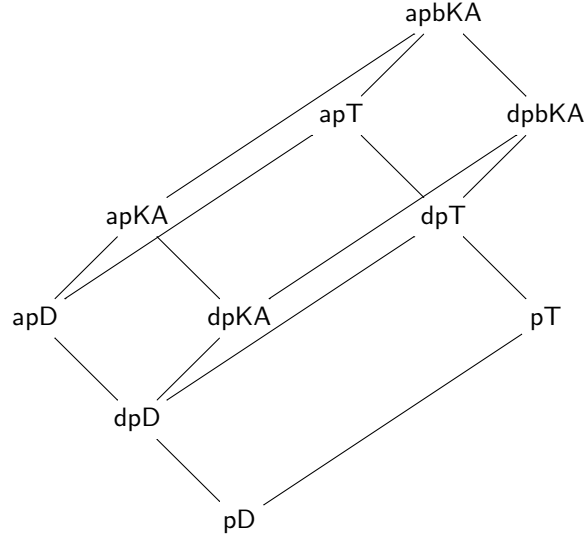
We have derived algebraic counterparts of Peleg's $\CDL$ axioms from
these two algebras. We have also proved their soundness with respect
to the concrete multirelational model.

The algebra of multirelations is, however, much richer than this
article might suggest. First of all, a left interaction law $R\cdot
(S\|T)\subseteq (R\cdot S)\|(R\cdot T)$ complements its dextrous
counterpart. Second, domain is characterised by the inclusion
$1_\sigma\cap R\cdot U \subseteq d(R)$, where $U$ is the universal
multirelation defined in Section~\ref{S:multirelations}, but an
equational definition $d(R)=1_\sigma\cap R\cdot U$ of domain, as in
the relational setting, is impossible. Third, sequentiality and
concurrency also interact via laws such as $1_\pi\cdot R = 1_\pi$ and
in particular $1_\pi\cdot \emptyset = 1_\pi$. In fact, whether a
multirelation $R$ satisfies $R\cdot \emptyset=\emptyset$, $R\cdot
\emptyset\neq \emptyset$, or even $R\cdot\emptyset=R$ depends on
whether or not pairs of the form $(a,\emptyset)$ occur in it. This
situation is similar to that of languages which contain finite and
infinite words. There one can define the finite part $\mathsf{fin}(L)$
and the infinite part $\mathsf{inf}(L)$ of a language $L$ and prove
laws such as $\mathsf{fin}(L)\cdot\emptyset =\emptyset$ and
$\mathsf{inf}(L)\cdot\emptyset = \mathsf{inf}(L)$. Here we can
consider the multirelations $\tau(R)=R\cap 1_\pi$ and
$\overline{\tau}(R)=R-1_\pi$, which satisfy $\tau(R)\cdot
\emptyset=\tau(R)$ and $\overline{\tau}(R)\cdot \emptyset =
\emptyset$, study the sets of these elements, and derive identities
for expressions such as $\tau(R\cdot S)$ or $\overline{\tau}(R\cup
S)$ in analogy to the language case. Elements
$(a,\emptyset)$ can be interpreted as modelling nontermination or
program errors; elements $\tau(R)$ can be seen as terminal elements,
since $\tau(R)\cdot S=\tau(R)$ holds for any multirelation $S$. A
detailed investigation is the aim of a successor paper.

While up-closed multirelations seem unsuitable for concurrency,
another subclass is interesting. Call a multirelation $R$
\emph{union-closed} if for all $a$ and $X\neq \emptyset$ the condition
$X\subseteq \{A\mid (a,A)\in R\}$ implies $(a,\bigcup X)\in R$. If $R$
has only finite internal nondeterminism, that is, for each $a$ there
are only finitely many $A$ with $(a,A)\in R$, then $R$ is union closed
if and only if $R||R\subseteq R$. It turns out that sequential
composition of union-closed multirelations is associative, while, in
contrast to the up-closed case, concurrent composition remains
nontrivial. In the context of concurrency it seems natural to require
that a multirelation can access the union of two separate sets from
some state whenever it can acces them individually. Adapting
concurrent dynamic algebras to union-closed relations is another
promising direction for future work. A further specialisation to
Parikh's game logic based on proto-Kleene algebras with domain and
antidomain seems another feasible restriction.

In conclusion, the results presented in this article lay the
foundation for a thourough algebraic exploration of Peleg's concurrent
dynamic logic with its extensions and variants, Parikh's game logics
and monotone predicate transformer semantics. Algebra has been
instrumental in taming the tedious syntactic manipulations at the
multirelational level in favour of first-order equational
reasoning. More succinct descriptions of the algebra of multirelations
will be given in sucessor papers. A unification of related approaches
to games and concurrency from this basis seems possible. The
integration of more advanced concepts such as communication,
synchronisation, knowledge or incentive constraints remains to be
explored.

\paragraph*{Acknowledgements}
  The authors acknowledge support by the Royal Society and JSPS
  KAKENHI grant number 25330016 for this research. They are grateful
  to Yde Vedema for drawing their attention to concurrent dynamic
  logic, and to Yasuo Kawahara, Koki Nishizawa, Toshinori Takai and
  Norihiro Tsumagari for enlightening discussions. The second author
  would like to thank the Department of Mathematics and Computer
  Science at Kagoshima University, where much of this work has been
  conducted, for its hospitality and the Department of Mathematics at
  Kyushu University for a pleasant short stay and financial support.

\bibliographystyle{plain}
\bibliography{cda}

%%%%%%%%%%%%%%%%%%%%%%%%%%%%%%%%%%%%%%%%%%%%%%%%%%%%%%%%%%%%%%%%%%%%%%%%%%%%%

\newpage

  \appendix

\section*{Appendix 1: Axioms of Concurrent Dynamic Algebras}
\label{A:axioms}

First we list the complete set of proto-trioid axionms.
  \begin{align*}
x+(y+z)&=(x+y)+z\\
x+y&=y+x\\
x+0&=x\\
x+x&=x\\
    1_\sigma \cdot x &= x\\
    x\cdot 1_\sigma &= x\\
x\cdot y + x\cdot z &\le x\cdot (y+z)\\
(x+y)\cdot z &= x\cdot z+ y\cdot z\\
0\cdot x &= 0\\
x\|(y\|z) &= (x\|y)\|z\\
x\|y &= y\|x\\
1_\pi\|x &=x\\
x\|(y+z)&=x\cdot y+x\cdot z\\
0\|x &=0
\end{align*}

Next we list the concurrent dynamics algebra axioms for distributive
lattices and boolean algebras. The left-hand column contains the
axioms of dp-bi-Kleene algebras, the right-hand column those of
ap-bi-Kleene algebras.\\

\begin{minipage}[h][7cm][t]{.45\textwidth}
\begin{align*}
d(x)\cdot (y\cdot z) &= (d(x)\cdot y)\cdot z\\
x\cdot (d(y)\cdot z) &=(x\cdot d(y))\cdot z\\
x\cdot (y\cdot d(z)) &= (x\cdot y)\cdot d(z)\\
x &\le d(x)\cdot x\\
d(x\cdot y) &= d(x\cdot d(y))\\
d(x+y) &=d(x)+d(y)\\
d(x)&\le 1_\sigma\\
d(0)&=0\\
(x\|y)\cdot d(z) &=(x\cdot d(z))\|(y\cdot d(z))\\
d(x\|y)&=d(x)\cdot d(y)\\ 
d(x)\|d(y)&=d(x)\cdot d(y)\\
1_\sigma + x\cdot x^\ast &\le x^\ast\\
d(z)+x\cdot y\le y &\Rightarrow x^\ast\cdot d(z)\le y
  \end{align*}
\end{minipage}
\begin{minipage}[h][7cm][t]{.45\textwidth}
  \begin{align*}
 a(x)\cdot (y\cdot z) &= (a(x)\cdot y)\cdot z\\
x\cdot (a(y)\cdot z) &=(x\cdot a(y))\cdot z\\
x\cdot (y\cdot a(z)) &= (x\cdot y)\cdot a(z)\\ 
a(x)\cdot x &=0\\
a(x\cdot y) &=a(x\cdot a(a(y))\\
a(x)+a(a(x)) &=1_\sigma\\
a(x)\cdot (y+z)&=a(x)\cdot y+a(x)\cdot z\\
(x\|y)\cdot a(z) &=(x\cdot a(z))\|(y\cdot a(z)\\
a(x\|y)&=a(x)+a(y)\\
a(x)\|a(y)&=a(x)\cdot a(y)\\
1_\sigma + x\cdot x^\ast &\le x^\ast\\
a(z)+x\cdot y\le y &\Rightarrow x^\ast\cdot a(z)\le y  
  \end{align*}
\end{minipage}

\vspace{\baselineskip}

\noindent To obtain dp-trioids and ap-trioids, the star axioms must be
dropped. To obtain proto-algebras, the concurrency axioms must be
dropped.

Finally, we show, for ap-bi-Kleene algebras, the definition of domain
from antidomain and those for the diamond and box operators.
  \begin{equation*}
    a(a(x)) = d(x)\qquad
\langle x\rangle y = d(x\cdot y)\qquad
[x]y = a(\langle x\rangle a(y))
  \end{equation*}

\newpage

%%%%%%%%%%%%%%%%%%%%%%%%%%%%%%%%%%%%%%%%%%%%%%%%%%%%%%%%%%%%%%%%%%%%%%%%%%%

\section*{Appendix 2: Proofs}
\label{A:proofs}

\vspace{\baselineskip}
\textsc{Proof of Lemma}~\ref{P:seqlaws}
  \begin{enumerate}
  \item The two facts follow directly from the definition of
    sequential composition.
  \item By definition, $(a,A)\not\in \emptyset$ for all $a\in X$ and
    $A\subseteq X$, hence $\emptyset\cdot R=\emptyset$.
\item 
   \begin{align*}
     &(a,A) \in (R\cdot S)\cdot T\\
     & \Leftrightarrow \exists B,C.\ (a,C) \in R \wedge \exists g.\ G_g(C)\subseteq S\wedge B = \bigcup g(C)
     \wedge \exists f.\ G_f(B) \subseteq T \wedge A = \bigcup f(B)\\
     &\Leftrightarrow \exists C.\ (a,C) \in R \wedge \exists g.\ G_g(C)\subseteq S 
     \wedge \exists f.\ G_f(\bigcup g(C))\subseteq T \wedge A = \bigcup_{c\in C}\bigcup_{x \in g(c)}f(x)\\
     &\Leftrightarrow \exists C.\ (a,C) \in R \wedge \exists f, g.\ (\forall c\in C.\ G_g(c)\in S 
     \wedge G_f(g(c)) \subseteq T) \wedge A = \bigcup_{c\in C}\bigcup_{x \in g(c)}f(x)\\
     &\Leftrightarrow \exists C.\ (a,C) \in R \wedge \exists f.\ \forall c\in C.\exists D.\ (c,D)\in S 
     \wedge G_f(D) \subseteq T \wedge A = \bigcup_{c\in C}\bigcup_{d \in D}f(d)\\
% &\Rightarrow \exists C.\ (a,C) \in R \wedge \exists h,g.\ (\forall c\in C.\ (c,g(c))\in S \\
% &\quad \wedge \forall x\in g(c).\ (x,h(x,c)) \in T) \wedge A = \bigcup_{c\in C}\bigcup_{x \in g(c)}h(x,c)\\
&\Rightarrow \exists C.\ (a,C) \in R \wedge \exists h.\ (\forall c\in C.\exists D.\ (c,D)\in S) \\
&\quad \wedge (\forall d\in D.\ (d,h(d,c)) \in T) \wedge A = \bigcup_{c\in C}\bigcup_{d \in D}h(d,c)\\
&\Rightarrow \exists C.\ (a,C) \in R \wedge \exists f,h.\ \forall c\in C.\exists D.\ (c,D)\in S \\
&\quad \wedge \forall d\in D.\ (d,h(d,c)) \in T \wedge f(c) = \bigcup_{d\in D}h(d,c) \wedge A = \bigcup_{c\in C}f(c)\\
&\Rightarrow \exists C.\ (a,C) \in R \wedge \exists f.\ \forall c\in C.\exists D.\ (c,D)\in S \\
&\quad \wedge \exists g.\ G_g(D) \subseteq T\wedge f(c) = \bigcup_{d\in D}g(d)) \wedge A = \bigcup_{c\in C}f(c)\\
&\Leftrightarrow (a,A) \in R\cdot (S\cdot T).
  \end{align*}
\item 
  \begin{align*}
    (a,A) \in (R\cup S)\cdot T &\Leftrightarrow \exists B.\ (a,B)\in R\cup S\wedge \exists f. \ G_f(B)\subseteq T\wedge A =\bigcup f(B)\\
 &\Leftrightarrow (\exists B.\ (a,B)\in R\wedge \exists f. \ G_f(B)\subseteq T\wedge A =\bigcup f(B))\\
 &\quad\vee (\exists B.\ (a,B)\in S\wedge \exists f. \ G_f(B)\subseteq T\wedge A =\bigcup f(B))\\
&\Leftrightarrow (a,A) \in R\cdot T \vee (a,A) \in S\cdot T\\
&\Leftrightarrow (a,A)\in R\cdot T \cup R\cdot T.
  \end{align*}
%\item Similar to the above proof.
\item We show that $R\cdot S \subseteq R\cdot (S\cup T)$. The claim
  then follows by symmetry and properties of least upper bounds.
  \begin{align*}
    (a,A) \in R\cdot S & \Leftrightarrow \exists B.\ (a,B)\in R\wedge \exists f.\ G_f(B)\subseteq S \wedge A=\bigcup f(B)\\
& \Rightarrow \exists B.\ (a,B)\in R\wedge \exists f.\ G_f(B)\subseteq S\cup T \wedge A=\bigcup f(B)\\
&\Leftrightarrow (a,A)\in R\cdot (S\cup T).
  \end{align*}
\end{enumerate}
$\Box$

\vspace{\baselineskip}
\textsc{Proof of Lemma}~\ref{P:conclaws}
\begin{enumerate}
\item
  \begin{align*}
    (a,A) \in (R\|S)\|T &\Leftrightarrow \exists B,C,D.\ A=B\cup C\cup D\wedge (a,B)\in R\wedge (a,C)\in S\wedge (a,D)\in T\\ 
&\Leftrightarrow (a,A)\in R\|(S\|T).
  \end{align*}
\item 
$    (a,A)\in R\|S \Leftrightarrow \exists B,C.\ A=B\cup C\wedge (a,B)\in R\wedge (a,C)\in S\Leftrightarrow (a,A)\in S\|R$.
\item Immediate from the definition of parallel
  composition and $1_\pi$.
\item Immediate from the definition of parallel composition.
  \item 
    \begin{align*}
      (a,A)&\in R\|(S\cup T)\\
&\Leftrightarrow \exists B,C.\ A=B\cup C\wedge (a,B)\in R\wedge ((a,C)\in S\vee (a,C)\in T)\\
&\Leftrightarrow\exists B, C.\ A=B\cup C\wedge  ((a,B)\in R\wedge ((a,C)\in S) \vee ((a,B)\in R\wedge ((a,C)\in T)\\
&\Leftrightarrow (a,A)\in R\|S\cup R\|T.
    \end{align*}
  \end{enumerate}
$\Box$

\vspace{\baselineskip}
\textsc{Proof of Lemma}~\ref{P:interaction}\\
    Since $G_f(A\cup B)\subseteq R\Leftrightarrow G_f(A)\subseteq
    R\wedge G_f(B)\subseteq R$, it follows that
  \begin{align*}
    (a,A) \in (R\|S)\cdot T &\Leftrightarrow \exists B,C.\  (a,B\cup C) \in R\|S \wedge \exists f.\ G_f(B\cup C)\subseteq T \wedge A = \bigcup f(B\cup C)\\
 &\Rightarrow \exists X,Y.\  A = X\cup Y\\
&\quad\wedge (\exists B.\ (a,B) \in R\wedge \exists f.\ G_f(B) \subseteq T \wedge X = \bigcup f(B))\\
&\quad\wedge (\exists C.\ (a,C) \in S\wedge \exists f.\ G_f(C) \subseteq T \wedge Y = \bigcup f(C))\\
&\Leftrightarrow (a,A) \in (R\cdot T)\|(S\cdot T).
  \end{align*}
$\Box$

\vspace{\baselineskip}
\textsc{Proof of Lemma}~\ref{P:subidinout}
   \begin{enumerate}
   \item Suppose $(a,A) \in R\cdot P$. Then there exists a set $B$
     such that $(a,B)\in R$ and, for all $b\in B$, $G_\iota(b)\in P$,
     and $A = \bigcup_{b\in B} \{b\}=B$. So $(a,A) \in R$ and
     $G_\iota(A)\subseteq P$.

     Suppose $(a,A)\in R$ and $G_\iota(a)\in P$ for all $a\in A$. Then
     $(a,A)\in R\cdot P$ by definition of sequential composition with
     $f=\iota$.
   \item Suppose $(a,A)\in P\cdot R$. Then $G_\iota(a) \in P$ and
     $(a,A) \in R$ by definition of sequential composition.  Suppose
     that $G_\iota(a)\in P$ and $(a,A) \in R$. Then $(a,A)\in P\cdot
     R$, using $f=\lambda x. A$.
   \end{enumerate}
$\Box$

\vspace{\baselineskip}
\textsc{Proof of Lemma}~\ref{P:subidlaws}
\begin{enumerate}
\item Let $R\subseteq 1_\sigma$. Then
      \begin{align*}
        (a,A) \in (R\cdot S)\cdot T
& \Leftrightarrow \exists B.\ G_\iota(a) \in R\ \wedge (a,B) \in S
 \wedge \exists f.\ G_f(B)\subseteq T \wedge A = \bigcup f(B)\\
&\Leftrightarrow G_\iota(a) \in R\ \wedge \exists B. (a,B) \in S
 \wedge \exists f.\ G_f(B)\subseteq T \wedge A = \bigcup f(B)\\
&\Leftrightarrow G_\iota(a)\in R \wedge (a,A)\in S\cdot T\\
&\Leftrightarrow (a,A)\in R\cdot (S\cdot T).
      \end{align*}
    Let $S\subseteq 1_\sigma$. Then
      \begin{align*}
        (a,A) &\in (R\cdot S)\cdot T\\
        &\Leftrightarrow \exists B.\ (a,B)\in R\wedge G_\iota(B) \subseteq S\wedge \exists f.\ G_f(B)\in T \wedge A = \bigcup f(B)\\
        &\Leftrightarrow \exists B.\ (a,B)\in R\wedge \exists f. \ G_\iota(B) \subseteq S\wedge G_f(B)\subseteq T \wedge A = \bigcup f(B)\\
        &\Leftrightarrow \exists B.\ (a,B)\in R\wedge \exists f. \ G_f(B) \subseteq S\cdot T\wedge A = \bigcup f(B)\\
        &\Leftrightarrow (a,A) \in R\cdot (S\cdot T).
      \end{align*}
     Let $T\subseteq 1_\sigma$. Then
      \begin{align*}
        (a,A) &\in (R\cdot S)\cdot T\\
 &\Leftrightarrow (a,A)\in R\cdot S\wedge G_\iota(A) \subseteq T\\
&\Leftrightarrow \exists B.\ (a,B)\in R\wedge \exists f.\ G_f(B) \subseteq S\wedge A = \bigcup f(B) \wedge G_\iota(A) \subseteq T\\
&\Leftrightarrow \exists B.\ (a,B)\in R\wedge \exists f.\ G_f(B) \subseteq S\wedge G_\iota(A) \subseteq T\wedge A = \bigcup f(B) \\
&\Leftrightarrow \exists B.\ (a,B)\in R\wedge \exists f.\ G_f(B) \subseteq S\wedge G_\iota(f(b)) \subseteq T \wedge A = \bigcup f(B) \\
&\Leftrightarrow \exists B.\ (a,B)\in R\wedge \exists f.\ G_f(B) \subseteq S\cdot T\wedge A = \bigcup f(B) \\
&\Leftrightarrow (a,A) \in R\cdot (S\cdot T).
      \end{align*}
  \item Let $P\subseteq 1_\sigma$.
  \begin{align*}
    (a,A)&\in (R\|S)\cdot P\\
    &\Leftrightarrow  \exists B,C.\ A=B\cup C\wedge (a,B)\in R\wedge (a,C)\in S \wedge G_\iota(A\cup B)\subseteq P\\
    &\Leftrightarrow  \exists B,C.\ A=B\cup C\wedge (a,B) \in R\wedge (a,C)\in S \wedge G_\iota(A)\subseteq P\wedge G_\iota(B)\subseteq P\\
    &\Leftrightarrow \exists B,C.\ A= B\cup C\wedge (a,B)\in R\cdot P\wedge (a,C)\in S\cdot P\\
    &\Leftrightarrow (a,A)\in (R\cdot P)\|(S\cdot P).
  \end{align*}
\item Let again $P\subseteq 1_\sigma$.
  \begin{align*}
    (a,A) \in P \cdot (R\cup S) &\Leftrightarrow G_\iota(a) \in P \wedge ((a,A)\in R \vee (a,A) \in S)\\
&\Leftrightarrow (G_\iota(a) \in P \wedge (a,A)\in R) \vee (G_\iota(a) \in P \wedge (a,A) \in S)\\
&\Leftrightarrow (a,A) \in P\cdot R\vee (a,A)\in P\cdot S\\
&\Leftrightarrow (a,A)\in P\cdot R\cup P\cdot S.
  \end{align*}
\end{enumerate}
$\Box$

\vspace{\baselineskip}
\textsc{Proof of Lemma}~\ref{P:domprops}
  \begin{enumerate}
   \item Obivous.
  \item $(a,A)\in d(R)\cdot R \Leftrightarrow G_\iota(a) \in d(R) \wedge
   (a,A) \in R\Leftrightarrow (a,A) \in R$.
   \item 
\begin{align*}
G_\iota(a) \in d(R\cup S) &\Leftrightarrow \exists B.\ (a,B) \in R \vee (a,B) \in S\\
&\Leftrightarrow \exists B.\ (a,B) \in R\vee \exists B. (a,B) \in S\\
&\Leftrightarrow G_\iota(a) \in d(R) \vee G_\iota(a)\in d(S)\\
&\Leftrightarrow G_\iota(a) \in d(R)\cup d(S).
\end{align*}

   \item Obvious from the definition of domain.
   \item 
     \begin{align*}
       G_\iota(a) \in d(R\cdot S) &\Leftrightarrow \exists B.\ (a,B)\in R\cdot S\\
&\Leftrightarrow \exists B,C.\ (a,C)\in R\wedge \exists f.\ G_f(C) \subseteq S\wedge B=\bigcup f(C)\\
&\Leftrightarrow \exists C.\ (a,C)\in R\wedge \exists f.\ G_f(C) \subseteq S\\
&\Leftrightarrow \exists C.\ (a,C)\in R\wedge G_\iota(C) \subseteq d(S)\\
&\Leftrightarrow \exists C.\ (a,C) \in R\cdot d(S)\\
&\Leftrightarrow G_\iota(a)\in d(R\cdot d(S)).
     \end{align*}
   \item
     \begin{align*}
       G_\iota(a) \in d(R\|S) &\Leftrightarrow \exists B.\ (a,B) \in R||S\\
&\Leftrightarrow \exists C,D.\ (a,C)\in R \wedge (a,D) \in S\\
&\Leftrightarrow \exists C.\ (a,C)\in R \wedge \exists D.\ (a,D) \in S\\
&\Leftrightarrow G_\iota(a)\in d(R) \wedge G_\iota(a) \in d(S)\\
&\Leftrightarrow G_\iota(a)\in d(R)\cap d(S).
     \end{align*}
\item Obvious.
   \end{enumerate}
$\Box$

\vspace{\baselineskip}
\textsc{Proof of Lemma}~\ref{P:antidomprops1}
  \begin{enumerate}
  \item Obviously, $(a,A) \in a(R)$ iff $A=\{a\}$ and $(a,A)\not\in
    d(R)$, which holds iff $(a,A) \in 1_\sigma$ and $(a,A)\not\in R$.
  \item $G_\iota(a) \in a(a(R)) \Leftrightarrow \neg\neg \exists A. (a,A) \in R \Leftrightarrow \exists A, (a,A) \in R \Leftrightarrow G_\iota(a)\in d(R)$.
  \item 
$G_\iota(a) \in d(a(R) ) \Leftrightarrow G_\iota(a) \in a(a(a(R))) \Leftrightarrow \neg\neg\neg \exists A.\ (a,A) \in R
 \Leftrightarrow G_\iota(a)\in a(R)$.
  \end{enumerate}
$\Box$

\vspace{\baselineskip}
\textsc{Proof of Lemma}~\ref{P:antidomprops2}
  \begin{enumerate}
  \item 
$      (a,A) \in a(R)\cdot R \Leftrightarrow G_\iota(a) \in a(R) \wedge (a,A)\in S\Leftrightarrow \neg\exists B. (a,B)\in R \wedge (a,A) \in R$
which is false.
\item $a(R\cdot S) = 1_\sigma \cap -d(R\cdot S) = 1_\sigma\cap -d(R\cdot d(S)) =
  a(R\cdot d(S))$.
\item $a(R)\cup d(R) = (1_\sigma\cap -d(R))\cup d(R)= 1_\sigma \cap (-d(R)\cup d(R)) = 1_\sigma \cap U = 1_\sigma$.
\item 
\begin{align*}
a(R\cup S) &= 1_\sigma\cap -d(R\cup S)\\
& = 1_\sigma\cap -(d(R)\cup d(S))\\
& = 1_\sigma\cap -d(R)\cap -d(S)\\
& = (1_\sigma\cap -d(R)\cap (1_\sigma\cap -d(S))\\
& = a(R)\cap a(S)\\
& = a(R)\cdot a(S).
\end{align*}
\item 
\begin{align*}
a(R\|S) &= 1_\sigma\cap -d(R\|S)\\
& = 1_\sigma\cap -(d(R)\cap d(S))\\
& = (1_\sigma \cap -d(R))\cup (1_\sigma\cap -d(S))\\
& = a(R)\cup a(S).
\end{align*}  
\item
    $a(R)\|a(S) = d(a(R))\|d(a(S)) = d(a(R))\cdot d(a(S)) = a(R)\cdot a(S)$.
  \end{enumerate}
$\Box$

\vspace{\baselineskip}
\textsc{Proof of Lemma}~\ref{P:protodomprops}
  \begin{enumerate}
  \item Immediate from additivity of domain.
  \item $x\le d(x)\cdot x$ is an axiom; $d(x)\cdot x \le x$ holds
    since $d(x)\le 1_\sigma$.
  \item $d(x\cdot y) = d(x\cdot d(y)) \le d(x\cdot 1_\sigma) = d(x)$.
  \item Let $x\le 1_\sigma$. Then $x=d(x)\cdot x\le d(x)\cdot 1_\sigma = d(x)$.
\item
  \begin{align*}
    d(d(x)\cdot y) &= d(d(d(x)\cdot y))\cdot d(d(x)\cdot y)\\
& = d(d(x)\cdot y)\cdot d(d(x)\cdot y)\\
& = d(d(x)\cdot d(y))\cdot d(d(x)\cdot y)\\
& \le d(d(x))\cdot d(y)\\
& = d(x)\cdot d(y), 
  \end{align*}
  using (1), (2) and (3). For the converse direction, $d(x)\cdot
  d(y)\le 1_\sigma$, and therefore $d(x)\cdot d(y)\le d(d(x)\cdot
  d(y))=d(d(x)\cdot y)$ by (4).
   \end{enumerate}
$\Box$

\vspace{\baselineskip}
\textsc{Proof of Lemma}~\ref{P:domprops2}\\
  We consider only the first property.  Let $x\le d(y)\cdot x$. Then
  \begin{equation*}
d(x)\le d(d(y)\cdot x)=d(y)\cdot d(x)\le d(x).
\end{equation*}

Let $d(x)\le d(y)$. Then $x=d(x)\cdot x\le d(y)\cdot x$.\quad $\Box$

\vspace{\baselineskip}
\textsc{Proof of Lemma}~\ref{P:dmrdomprops}
  \begin{enumerate}
  \item $1_\sigma= d(1_\sigma)=d(1_\sigma\|1_\pi) = d(1_\sigma)\cdot d(1_\pi)\le d(1_\pi)$. The
    converse direction is obvious.
  \item $d(d(x)\|d(y))=d(d(x)\cdot d(y))=d(x)\cdot d(y)=d(x)\|d(y)$ by
    meet closure.
  \item $d(x)\|d(x)=d(x)\cdot d(x)=d(x)$.
  \end{enumerate}
$\Box$

\vspace{\baselineskip}
\textsc{Proof of Lemma}~\ref{P:demodalisation}\\
 Suppose $\langle x\rangle p \le d(q)$, that is, $d(x\cdot p)\le
  d(q)$. Then, by Lemma~\ref{P:protodomprops}(2),
  \begin{equation*}
    x\cdot d(p)=d(x\cdot p)\cdot x\cdot d(p)\le d(q)\cdot x\cdot d(p) \le d(q)\cdot x.
  \end{equation*}
  For the converse implication, suppose $x\cdot d(p)\le d(q)\cdot
  x$. Then $(x\cdot d(p))\cdot d(p)\le (d(q)\cdot x)\cdot d(p)$ and
  therefore $x\cdot d(p)\le d(q)\cdot (x\cdot d(p))$ by domain
  associativity and idempotency. Hence 
\begin{equation*}
\langle x\rangle p = d(x\cdot
  p) \le d(d(q)\cdot x \cdot d(p)) = d(q)\cdot d(x\cdot d(p))\le d(q)=q
\end{equation*}
by isotonicity of domain, domain export and properties of meet.\quad $\Box$

\vspace{\baselineskip}
\textsc{Proof of Lemma}~\ref{P:fixpointprops}
  \begin{enumerate}
  \item The functions $\lambda X.\ R\cdot X$ and $\lambda X.\ S\cup X$
    are isotone for all $R$ and $S$, hence so are their compositions.
  \item Every ring of sets forms a complete lattice.
  \item This follows from (1) and (2) by standard fixpoint theory
    (Knaster-Tarski Theorem).
  \end{enumerate}
$\Box$

\vspace{\baselineskip}
\textsc{Proof of Lemma}~\ref{P:simulation}\\
  Suppose $x\cdot p\le p\cdot y$. For $x^\ast p\le p\cdot y^\ast$, it
  suffices to show that $p+x\cdot (p\cdot y^\ast) \le p\cdot y^\ast$
  by star induction. First, $p \le p\cdot y^\ast$ by left isotonicity
  of multiplication and star unfold. Moreover, by the assumption and
  domain associativity
\begin{equation*}
  x\cdot (p\cdot y^\ast) = (x\cdot p)\cdot y^\ast \le (p\cdot y)\cdot y^\ast = p\cdot (y\cdot y^\ast) \le p\cdot y^\ast.
\end{equation*}
$\Box$

\vspace{\baselineskip}
\textsc{Proof of Lemma}~\ref{P:starvar}
\begin{enumerate}
\item Obviously, $p\le \langle x^\ast\rangle p$ by the left unfold
  law. For $\langle x^\ast\rangle\langle x\rangle p \le \langle
  x^\ast\rangle p$ it suffices, by star induction, to show that
  $\langle x\rangle p \le \langle x^\ast\rangle p$ and $\langle
  x\rangle\langle x^\ast\rangle p \le \langle x^\ast\rangle p$.  The
  first inequality follows from $\langle 1_\sigma\rangle p \le \langle
  x^\ast\rangle p$ and $\langle x\rangle p = \langle x\rangle\langle
  1_\sigma\rangle p$ by isotonicity. The second one holds by left star
  unfold.
\item  Let $p\le q$ and $\langle x\rangle q\le q$. Hence $\langle
  x^\ast\rangle q\le q$ by Proposition~\ref{P:cdainduct} and the claim
  follows by domain isotoniticy.
\end{enumerate}
$\Box$

\vspace{\baselineskip}
\textsc{Proof of Lemma}~\ref{P:antidomprops}
Note that $d$ is an appbreviation of $a\circ a$.
  \begin{enumerate}
  \item Obvious from the third antidomain axiom.
  \item 
$a(x) = (a(x)+d(x))\cdot a(x)
=a(x)\cdot a(x)+a(a(x))\cdot a(x)
=a(x)\cdot a(x)+0
 =a(x)\cdot a(x)$.
  \item This holds since $a(1_\sigma)=0$ and $1_\sigma\cdot x = 0$ implies
    $x=0$.
  \item Let $a(x)\le a(y)$. Then $a(x)\cdot y \le a(y)\cdot y = 0$.

For the converse direction,
\begin{equation*}
  a(x)\cdot y = 0 \Leftrightarrow a(a(x)\cdot y) = 1_\sigma\Leftrightarrow a(a(x)\cdot d(y)) = 1_\sigma\Leftrightarrow a(x)\cdot d(y)=0.
\end{equation*}
and therefore
\begin{equation*}
  a(x)=a(x)\cdot (d(y)+a(y))=a(x)\cdot d(y)+a(x)\cdot a(y)= a(x)\cdot a(y)\le a(y).
\end{equation*}
\item $a(y)\cdot x \le a(y)\cdot y=0$, so $a(y)\le a(x)$ by (3).
\item 
  \begin{equation*}
    a(x)\cdot a(y)\cdot (x+y) = a(x)\cdot a(y)\cdot x+ a(x)\cdot a(y)\cdot y \le a(x)\cdot x+a(y)\cdot y=0.
  \end{equation*}
Moreover, by (4),
\begin{align*}
  a(x)\cdot a(y) \cdot (x+y)= 0 &\Leftrightarrow a(a(x)\cdot a(y) \cdot (x+y))= 1_\sigma\\
&\Leftrightarrow a(a(x)\cdot a(y) \cdot d(x+y))= 1_\sigma\\
&\Leftrightarrow a(x)\cdot a(y) \cdot d(x+y)= 0.
\end{align*}

\item 
 $    a(x+y)\le a(x)$ and $a(x+y)\le a(y)$ by (5), so 
\begin{equation*}
a(x+y)=a(x+y)\cdot a(x+y)\le a(x)\cdot a(y).
\end{equation*}
 For the converse direction, by (6),
  \begin{align*}
    a(x)\cdot a(y)
& = a(x)\cdot a(y)\cdot a(x+y)+a(x)\cdot a(y)\cdot d(x+y)\\
&=a(x)\cdot a(y)\cdot a(x+y) \le a(x+y).
  \end{align*}
\item First $a(y)\le a (a(x)\cdot y)$ and $d(x)\le a (a(x)\cdot y)$ by
  antitonicity, so 
\begin{equation*}
d(x) + a(y) \le a(a(x)\cdot y)
\end{equation*}
by properties of least upper bounds.

  For the converse direction, we have $a(a(x)\cdot y)\cdot a(x)\cdot
  d(y)=0$. Therefore,
\begin{align*}
  a(a(x)\cdot y) & = a(a(x)\cdot y)\cdot d(y)+a(a(x)\cdot y)\cdot a(y)\\
&\le a(a(x)\cdot y)\cdot d(y)+a(y)\\
&= a(a(x)\cdot y)\cdot a(x)\cdot d(y)+a(a(x)\cdot y)\cdot d(x)\cdot d(y)+a(y)\\
&= a(a(x)\cdot y)\cdot d(x)\cdot d(y)+a(y)\\
&\le d(x)+a(y).
\end{align*}
  \end{enumerate}
$\Box$

\vspace{\baselineskip}
\textsc{Proof of Proposition}~\ref{P:antidomaindomain}\\
We verify the domain axioms in the setting of ap-dioids.
\begin{itemize}
\item The associativity laws
  \begin{align*}
    d(x)\cdot (y\cdot z)&=(d(x)\cdot y)\cdot z,\qquad x\cdot (d(y)\cdot z)\\
&=(x\cdot d(y))\cdot z,\qquad x\cdot (y\cdot d(z))\\
&= (x\cdot y)\cdot d(z)
  \end{align*}
  are immediate from antidomain associativity.

\item $d(x)\le 1_\sigma$ is immediate from the complementation axiom.

\item $d(x)\cdot x = x$ holds because $x=(d(x)+a(x))\cdot x =
  d(x)\cdot x + 0$ by the complementation and left annihilation axiom.

\item $d(x\cdot y)=d(x\cdot d(y))$ is immediate from antidomain
  locality.

\item $d(0)=0$ holds because $a(0)=1_\sigma$ and $a(1_\sigma)=0$.

\item $d(x+y)=d(x)+d(y)$ holds because, by antidomain multiplicativity
  and export,
\begin{equation*}
d(x+y)=a(a(x+y))=a(a(x)\cdot a(y))=d(x)+a(a(y))=d(x)+d(y).
\end{equation*}
\end{itemize}
$\Box$

\vspace{\baselineskip}
\textsc{Proof of Proposition}~\ref{P:amrdmr}\\
Every ap-dioid is a dp-dioid by
Proposition~\ref{P:antidomaindomain}. We verify the remaining axioms
for parallel composition.
  \begin{itemize}
  \item The domain interaction axiom $(x\cdot d(z))\| (y\cdot d(z))
    = (x\|y)\cdot d(z)$ follows immediately from the antidomain
    interaction axiom.
  \item $d(x||y)=d(x)\cdot d(y)$ holds because
    \begin{equation*}
      d(x\|y)=a(a(x\|y))= a(a(x)+a(y))=a(a(x))\cdot a(a(y))=d(x)\cdot d(y),
    \end{equation*}
    using the De Morgan law for $a$ and the first antidomain
    concurrency axiom.
  \item $d(x)\|d(y)=d(x)\cdot d(y)$ is immediate from the
   second antidomain concurrency axiom.
  \end{itemize}
$\Box$

\vspace{\baselineskip}
\textsc{Proof of Lemma}~\ref{P:chaincompletecounter}\\
  Let $R=\{(n,\mathbb{N})\mid n\in\mathbb{N}\}$ and $S_i=\{(n,\{m\})\mid n\in\mathbb{N}\wedge 0\le m\le i\}$. Thus clearly $S_i \subset
  S_j$ whenever $i<j$. Moreover,
  \begin{align*}
    (n,A) \in R\cdot \bigcup_{i\in\mathbb{N}}S_i &\Leftrightarrow (n,\mathbb{N}) \in R \wedge \exists f.\ G_f(\mathbb{N}) \subseteq \bigcup_{i\in\mathbb{N}}S_i \wedge A = \bigcup_{n\in\mathbb{N}} f(n)\\
&\Leftrightarrow (n,\mathbb{N}) \in R \wedge \exists m\in\mathbb{N}.\ (\forall n\in\mathbb{N}.\ (n,\{m\}) \in \bigcup_{i\in\mathbb{N}}S_i) \wedge A = \bigcup_{n\in\mathbb{N}} \{n\}\\
&\Leftrightarrow (n,\mathbb{N}) \in R \wedge \exists m\in\mathbb{N}.\ (\forall n\in\mathbb{N}.\ (n,\{m\}) \in \bigcup_{i\in\mathbb{N}}S_i) \wedge A = \mathbb{N}
  \end{align*}
  and therefore $(n,\mathbb{N}) \in F_R(\bigcup_{i\in\mathbb{N}}R_i)$
  for all $(n,\mathbb{N})$. However, 
  \begin{align*}
    (n,A) \in R\cdot S_i \Leftrightarrow (n,\mathbb{N}) \in R \wedge \exists m\le i. (\forall n\in\mathbb{N}.(n,\{m\}) \in R_i \wedge A=\bigcup_{0\le k\le i}\{k\},
  \end{align*}
hence no $F_R(R_i)$ contains
  $(n,\mathbb{N})$ for any $n$ and therefore also not the union
  $\bigcup_{i\in\mathbb{N}}F_R(R_i)$.\quad $\Box$

\vspace{\baselineskip}
\textsc{Proof of Lemma}~\ref{P:chaincont}\\
  Suppose a family $\{S_i\mid i\in\mathbb{N}\}$ such that $S_i\subset
  S_j$ whenever $i<j$. We must show that
  $F_R(\bigcup_{i\in\mathbb{N}}S_i) \subseteq \bigcup_{i\in\mathbb{N}}
  F_R(S_i)$. So suppose $(a,A) \in
  F_R(\bigcup_{i\in\mathbb{N}}S_i)$. If $(a,A) \in 1_\sigma$, then $(a,A)
  \in \bigcup_{i\in\mathbb{N}} F_R(S_i)$.

  Otherwise, if $(a,A) \in R\cdot \bigcup_{i\in\mathbb{N}}R_i$, then
  there is a finite set $B=\{b_1,\dots b_k\}$ and there are sets
  $A_1,\dots,A_k$ such that $(a,B)\in R$, all $(b_i,A_i) \in
  \bigcup_{i\in\mathbb{N}}R_i$ and
  $A=\bigcup_{i\in\mathbb{N}}A_i$. Hence for all $1\le i\le k$ there
  exists a $l_i$ such that $(b_i,A_i)\in S_{l_i}$. Because of the
  ascending chain condition there exists a maximal $S_m$ such that all
  $(b_i,A_i)\in S_m$. Then $(a,A) \in F_R(S_m)$ and finally also
  $(a,A) \in \bigcup_{i\in\mathbb{B}}F_R(S_i)$.\quad $\Box$

\vspace{\baselineskip}
\textsc{Proof of Lemma}~\ref{P:iterzeroone}
\begin{enumerate}
\item  In the base case, $F_R^0(\emptyset)=\emptyset = R^{(0)}$. In the induction step,
  \begin{equation*}
    F_R^{(n+1)}(\emptyset) = F_R(F_R^n(\emptyset)) = 1_\sigma\cup R\cdot R^{(n)} = R^{(n+1)}.
  \end{equation*}
Finally, $F_R^\ast(\emptyset) = \bigcup_{n\in\mathbb{N}} F_R^n(\emptyset) = \bigcup_{n\in\mathbb{N}} R^{(n)} = R^{(\ast)}$.
\item Immediate from (1).
\end{enumerate}
 $\Box$

\vspace{\baselineskip}
\textsc{Proof of Lemma}~\ref{P:segerbergconv}\\
$p + \langle x^\ast\rangle (\langle x\rangle p - p) \le p +\langle x^\ast\rangle \langle x \rangle p \le \langle x^\ast \rangle p$, by Lemma~\ref{P:starvar}(1).\quad $\Box$

%%%%%%%%%%%%%%%%%%%%%%%%%%%%%%%%%%%%%%%%%%%%%%%%%%%%%%%%%%%%%

\section*{Appendix 3: Proof Automation with Isabelle/HOL}
\label{A:isabelle}

Some of the proofs at the multirelational level in this article are
technically tedious, in partiular those using second-order
Skolemisation. Reasoning algebraically about domain and antidomain in
the absence of associativity of sequential composition is intricate
for different reasons. We have therefore formalised the mathematical
structures used in this article and verified many of our proofs with the
interactive proof assistant Isabelle/HOL~\cite{NipkowPW02}. In
particular, the complete technical development in this article from
multirelations to star-free concurrent dynamic algebras and the
complete algebraic layer have been formally verified. Finally,
Isabelle's built-in counterexample generators Quickcheck and Nitpick
have helped in finding some counterexamples.

We now list in detail the facts which have and have not been formally
verified.
\begin{description}
\item[Section~\ref{S:mulproperties}] We have verified
  Lemma~\ref{P:seqlaws}, except for part (3), which is not needed for
  our results, the isotonicity properties of sequential composition,
  Lemma~\ref{P:conclaws}, isotonicity of concurrent composition and
  Lemma~\ref{P:interaction}. Isabelle also provided the
  counterexamples in Lemma~\ref{P:counterexamples}.
\item[Section~\ref{S:subidlaws}] All statements
  (Lemma~\ref{P:subidinout} and~\ref{P:subidlaws}) have been
  verified. The subalgebra of subidentities has not been formalised.
\item[Section~\ref{S:multireldom}] All statements, Lemma~\ref{P:domprops} to
  Corollary~\ref{P:domassocinter}, have been verified.
\item[Section~\ref{S:axioms}] We have verified irredundancy of the
  domain and antidomain axiom sets of domain and antidomain
  proto-dioids and proto-trioids.  We have not explicitly formalised
  Theorem~\ref{P:dmramrmodel}, but all facts needed in the proof have
  been verified.
\item[Section~\ref{S:dmrdioids}] Lemma~\ref{P:domretraction} has been
  verified, but not Proposition~\ref{P:retractionlemma}, which is a
  well known consequence. Lemma~\ref{P:protodomprops} and the
  individual equational proof steps for Proposition~\ref{P:domaindl}
  have been verified; the precise statement of
  Proposition~\ref{P:domaindl} has not been
  formalised. Lemma~\ref{P:domprops2} and Lemma~\ref{P:dmrdomprops}
  have been verified. The remaining facts in this section
  (Proposition~\ref{P:compprop} to Theorem~\ref{P:greatestba}) have
  not been verified.
\item[Section~\ref{S:diaaxioms}] All proofs and counterexamples,
  Lemma~\ref{P:cdlaxioms1} to~\ref{P:diacounter}, have been verified.
\item[Section~\ref{S:star}] Lemma~\ref{P:fixpointprops} to
  Theorem~\ref{P:dmrkamodel} have not been verified; formalising the
  underlying concepts seems excessive relative to the moderate
  difficulty of proofs. Lemma~\ref{P:cdaunfold} to
  Lemma~\ref{P:starvar} have been verified. Theorem~\ref{P:cda}, which
  combines these results, as not been formalised as such.
\item[Section~\ref{S:amrdioids}] Lemma~\ref{P:antidomprops} has been
  verified. All the equational proof steps for
  Proposition~\ref{P:antidomaindomain},
  Proposition~\ref{P:antidomainba} and
  Proposition~\ref{P:KAdomainantidomain} have been verified, but the
  individual statements have not been
  formalised. Proposition~\ref{P:antidomainba} has not been verified.
  Theorem~\ref{P:dmrkamodel} has not been verified, because the star
  in the multirelational model has not been formalised.
\item[Section~\ref{S:boxaxioms}] Lemma~\ref{P:antidomaindiacda} and
  Proposition~\ref{P:antidomainboxcda} have been
  verified. Theorem~\ref{P:cdaboxdiamond} has not been formalised as,
  but individual proof steps have been
  verified. Lemma~\ref{P:boxcounter} has not been verified because it
  holds by duality between box and diamonds.
\item[Section~\ref{S:finiteiteration}] No results have been verified.
\item[Section~\ref{S:segerberg}] No results have been verified.
\end{description}

As mentioned in the Introduction, the complete Isabelle development
with all proofs listed above can be found online.

%%%%%%%%%%%%%%%%%%%%%%%%%%%%%%%%%%%%%%%%%%%%%%%%%%%%%%%%%%%%%%%%%%%%%%

\end{document}